%% file: paper.tex
\documentclass[runningheads,final]{llncs}
\pdfoutput=1 
\usepackage[utf8]{inputenc}
\usepackage{amsmath,amssymb}
\usepackage{amsthm}
\usepackage{thmtools, thm-restate}
\usepackage{graphicx}
\usepackage{etoolbox}
\usepackage{cite}
\usepackage{mathpartir}
\usepackage{mathtools}
\usepackage{proof}
\usepackage[most]{tcolorbox}
\usepackage{turnstile}
\usepackage{url}
\usepackage{hyperref}
\usepackage{xcolor}
\usepackage{xargs}
\usepackage{paralist}
\usepackage{listings}
\usepackage{caption}
\usepackage{subcaption}
\usepackage{xfrac}
\usepackage{diagbox}
\usepackage{makecell}
\usepackage{arydshln}
\usepackage{csquotes}
\usepackage{multirow}
\usepackage{wrapfig}
\usepackage{ebproof}
\usepackage{ifthen}
\usepackage[capitalise]{cleveref}

\usepackage{paper}

\newcommand{\citeappendix}{\cite{extended}}

\newboolean{appendix}

\setboolean{appendix}{true}

\newcommand{\ifappendixthenelse}[2]{\ifthenelse{\boolean{appendix}}{#1}{#2}}

\newcommand{\aref}[1]{\ifappendixthenelse{\cref{#1}}{\citeappendix}}

\newcommand{\aword}{\ifappendixthenelse{the Appendix}{\citeappendix}}

\definecolor{darkgray}{rgb}{0.31,0.31,0.33}
\definecolor[named]{warmYellow}{rgb}{0.99,0.78,0.07}

\definecolor[named]{keyword}{HTML}{e0004a}
\definecolor[named]{constructor}{HTML}{009304}
\definecolor[named]{extension}{HTML}{1e88e5}
\definecolor[named]{basic}{HTML}{1b0160}

\newcommand*{\mycommentstyle}[1]{\color{darkgray}\rmfamily\itshape\lstset{columns=fullflexible}{#1}}

\lstdefinelanguage{splay}{%
  basicstyle=\color{basic}\ttfamily\lst@ifdisplaystyle\scriptsize\else\small\fi,
  extendedchars=true,
  alsoletter={*,'},
  keywordstyle=[1]\color{constructor},
  keywordstyle=[2]\color{keyword},
  keywordstyle=[3]\color{extension},
  keywords=[1]{true,false,leaf,node},
  keywords=[2]{if,then,else,match,with,let,in},
  keywords=[3]{coin,nondet},
  sensitive=true,
  moredelim=[is][\mycommentstyle]{(*}{*)},
  mathescape=true,
  numberstyle=\tiny
}

\tcbuselibrary{skins,breakable}

\tcbset{colback=white,
  colframe=warmYellow,
  highlight math style= {enhanced, 
    colframe=red,colback=red!10!white,boxsep=0pt},
  arc=0mm,
  boxrule=.2mm,
}

\newtcblisting{broken}{
      arc=0mm,
      top=-1mm,
      bottom=-1mm,
      left=0mm,
      right=0mm,
      width=\textwidth,
      boxrule=.2mm,
      colback=white,
      listing only,
      listing options={},
      breakable
}

\lstnewenvironment{code}{}{}

\begin{document}
\title{Automated Expected Amortised Cost Analysis of Probabilistic Data Structures}
\titlerunning{Automated Expected Amortised Cost Analysis}
\author{Lorenz Leutgeb\inst{1}
\and
Georg Moser\inst{2}
\and 
Florian Zuleger\inst{3}
}
\authorrunning{L.~Leutgeb et al.}
\institute{
  Max Planck Institute for Informatics and Graduate\\
  School of Computer Science, 
  Saarbrücken, Germany
  \and
  Department of Computer Science\\
  Universität Innsbruck, Austria
  \and
  Institute of Logic and Computation 192/4\\
  Technische Universität Wien, Austria
}
\maketitle

\begin{abstract}
In this paper, we present the first fully-automated \emph{expected amortised cost analysis} of self-adjusting data structures, that is, of \emph{randomised splay trees}, \emph{randomised splay heaps} and \emph{randomised meldable heaps}, which so far have only (semi-) manually been analysed in the literature.
Our analysis is stated as a type-and-effect system for a first-order functional programming language with support for sampling over discrete distributions, non-deterministic choice and a ticking operator. The latter allows for the specification of fine-grained cost models.
We state two soundness theorems based on two different---but strongly related---typing rules of ticking, which account differently for the cost of non-terminating computations.
Finally we provide a prototype implementation able to fully automatically analyse the aforementioned case studies.

\keywords{amortised cost analysis \and
functional programming \and
probabilistic data structures \and
automation \and
constraint solving}
\end{abstract}

\section{Introduction}
\label{sec:intro}
\input{introduction.tex}

\section{Overview of Our Approach and Results}
\label{sec:overview}
\input{overview.tex}

\section{Probabilistic Functional Language}
\label{sec:language}
\input{language.tex}

\section{Operational Semantics}
\label{sec:semantics}
\input{semantics.tex}

\section{Type-and-Effect System for Expected Cost Analysis}
\label{sec:typesystem}
\input{typesystem.tex}

\section{Implementation and Evaluation}
\label{sec:implementation}
\input{implementation.tex}

\section{Conclusion}
\label{sec:conclusion}
\input{conclusion.tex}

\newpage

\ifthenelse{\boolean{appendix}}{%
Appendix
\newpage
\appendix
\counterwithin{figure}{section}
\counterwithin{table}{section}
\input{appendix.tex}

}{}
\end{document}

%% file: introduction.tex
\emph{Probabilistic} variants of well-known computational models such as automata, Turing machines or the $\lambda$-calculus have been studied since the early days of computer
science (see~\cite{Kozen81,Kozen:JCSC:85,MR:1999} for early references).
One of the main reasons for considering probabilistic models is that they often allow for the design of more efficient algorithms than their deterministic counterparts (see  e.g.~\cite{MR:1999,MitzenMacherU05,BKS:2020}).
Another avenue for the design of efficient algorithms has been opened up by Sleator and Tarjan~\cite{ST:1985,Tarjan:1985} with their introduction of the notion of \emph{amortised complexity}.
Here, the cost of a single data structure operation is not analysed in isolation but as part of a sequence of data structure operations.
This allows for the design of algorithms where the cost of an expensive operation is averaged out over multiple operations and results in a good overall \emph{worst-case cost}.
Both  methodologies---\emph{probabilistic programming} and \emph{amortised complexity}---can be combined for the design of even more efficient algorithms, as for example in \emph{randomized splay trees}~\cite{AlbersK02}, where a rotation in the splaying operation is only performed with some probability (which improves the overall performance by skipping some rotations while still guaranteeing that enough rotations are performed).

In this paper, we present the first fully-automated \emph{expected amortised cost analysis} of probabilistic data structures, that is, of \emph{randomised splay trees}, \emph{randomised splay heaps}, \emph{randomised meldable heaps} and a \emph{randomised analysis} of a \emph{binary search tree}.
These data structures have so far only (semi-)manually been analysed in the literature.
Our analysis is based on a novel type-and-effect system, which constitutes a generalisation of the type system studied in~\cite{hofmann2021typebased,LMZ:2021} to the non-deterministic and probabilistic setting, as well as an extension of the type system introduced in~\cite{WangKH20} to sublinear bounds and non-determinism. We provide a prototype implementation that is able to fully automatically analyse the case studies mentioned above.
We summarize here the main contributions of our article:
\begin{inparaenum}[(i)]
\item We consider a first-order functional programming language with support for \emph{sampling over discrete distributions},
    \emph{non-deterministic choice} and a \emph{ticking} operator, which allows for the specification of fine-grained cost models.
\item We introduce compact \emph{small-step} as well as \emph{big-step} semantics for our  programming language.
    These semantics are equivalent wrt. the obtained normal forms (ie., the resulting probability distributions) but differ wrt. the cost assigned to non-terminating computations.
\item Based on~\cite{hofmann2021typebased,LMZ:2021}, we develop a novel type-and-effect system that strictly generalises the prior approaches from the literature.
\item We state two soundness theorems (see Section~\ref{Soundness}) based on two different---but strongly related---typing rules of ticking.
    The two soundness theorems are stated wrt. the small-step resp. big-step semantics because these semantics precisely correspond to the respective ticking rule.
    The more restrictive ticking rule can be used to establish (positive) almost sure termination (AST) while the more permissive ticking rule supports the analysis of a larger set of programs (which can be very useful in case termination is not required or can be established by other means);
    in fact, the more permissive ticking rule is essential for the precise cost analysis of randomised splay trees.    
    We note that the two ticking rules and corresponding soundness theorems do not depend on the details of the type-and-effect system, and we believe that they will be of independent interest (e.g., when adapting the framework of this paper to other benchmarks and cost functions).
\item Our prototype implementation~\atlas\ strictly extends and earlier vrsion discussed in~\cite{LMZ:2021} and all our earlier evaluation results can be replicated (and sometime improved).
\end{inparaenum}

With our implementation and the obtained experimental results we make two contributions to the complexity analysis of data structures:
\begin{enumerate}
\item \emph{We automatically infer bounds on the expected amortised cost, which could previously only be obtained by sophisticated pen-and-paper proofs. 
    In particular, we verify that the amortised costs of randomised variants of self-adjusting data structures improve upon their non-randomised variants.}
  In Table \ref{tab:results} we state the expected cost of the randomised data structures and their deterministic counterparts; the benchmarks are detailed in Section~\ref{sec:overview}.

\item \emph{We establish a novel approach to the expected cost analysis of data structures.}
  %
  While the detailed study of Albers et al.\ in~\cite{AlbersK02} requires a sophisticated pen-and-paper analysis,
  our approach allows us to fully-automatically compare the effect of different rotation probabilities on the expected cost (see Table~\ref{tab:matrix} of Section~\ref{sec:implementation}).
\end{enumerate}

\begin{table}[t]
\setlength{\tabcolsep}{0.5em}
\setlength\dashlinedash{0.5pt}
\setlength\dashlinegap{1.0pt}
\centering
\begin{tabular}{|l:l:l|}
\hline
& probabilistic & deterministic~\cite{LMZ:2021} \\
\hline\hline
\multicolumn{3}{|c|}{Splay Tree} \\
\hdashline
\flst{insert} &
$\sfrac{3}{4} \log(|t|) + \sfrac{3}{4} \log(|t|+1)$ &
$2 \log(|t|) + \sfrac{3}{2}$ \\
\flst{delete} &
$\sfrac{9}{8} \log(|t|)$ &
$\sfrac{5}{2} \log(|t|) + 3$ \\
\flst{splay} &
$\sfrac{9}{8} \log(|t|)$ &
$\sfrac{3}{2} \log(|t|)$ \\
\hline\hline
\multicolumn{3}{|c|}{Splay Heap} \\
\hdashline
\flst{insert} &
$\sfrac{3}{4} \log(|h|) + \sfrac{3}{4} \log(|h|+1)$ &
$\sfrac{1}{2} \log(|h|) + \log(|h|+1) + \sfrac{3}{2}$ \\
\flst{delete\_min}  &
$\sfrac{3}{4} \log(|h|)$ &
$\log(|h|)$ \\
\hline\hline
\multicolumn{3}{|c|}{Meldable Heap} \\
\hdashline
\flst{insert} &
$\log(|h|) + 1$ &\\
\flst{delete\_min} &
$2 \log(|h|)$ &
\multicolumn{1}{c|}{\emph{not applicable}}\\
\flst{meld} &
$\log(|h_1|) + \log(|h_2|)$ &\\
\hline\hline
\multicolumn{3}{|c|}{Coin Search Tree} \\
\hdashline
\flst{insert} &
$\sfrac{3}{2} \log(|t|) + \sfrac{1}{2}$ &\\
\flst{delete} &
$\log(|t|)$ &
\multicolumn{1}{c|}{\emph{not applicable}}\\
\flst{delete\_max} &
$\sfrac{3}{2} \log(|t|) + \sfrac{1}{2}$ &\\
\hline
\end{tabular}
\vspace{3mm}
\caption{Expected Amortised Cost of Randomised Data Structures. We also state the deterministic counterparts considered in~\cite{LMZ:2021} for comparison.}
\vspace{-5mm}
\label{tab:results}
\end{table}

\paragraph*{Related Work.}
The generalisation of the model of computation and the study of the expected resource usage of
\emph{probabilistic} programs has recently received increased attention (see e.g.~\cite{CFM:CAV:17,McIverMKK18,KKMO:ACM:18,NCH:2018,BatzKKMN19,WFGCQS:PLDI:19,ADY:2020,AMS:2020,EberlHN20,WangKH20,MeyerHG21,ABD:2021,MoosbruggerBKK21}).
%
We focus on related work concerned with automations of expected cost analysis of
deterministic or non-determi\-nistic, probabilistic programs---imperative or functional.
(A probabilistic program is called \emph{non-deterministic}, if it additionally makes use of non-deterministic
choice.)

In recent years the \emph{automation} of expected cost analysis of probabilistic data structures or programs has
gained momentum, cf.~\cite{NCH:2018,AvanziniLG19,WFGCQS:PLDI:19,ADY:2020,AMS:2020,WangKH20,MeyerHG21,ABD:2021,MoosbruggerBKK21}.
Notably, the \Absynth\ prototype by~\cite{NCH:2018}, implement Kaminski's \ert-calculus, cf.~\cite{KKMO:ACM:18} for reasoning about expected costs.
Avanzini et al.~\cite{AMS:2020} introduce the tool~\ecoimp, which generalises
the \Absynth\ prototype and provides a modular and thus a more efficient and scalable alternative
for non-determi\-nistic, probabilistic programs.
In comparison to these works, we base our analysis on a dedicated type system finetuned
to express sublinear bounds; further our prototype implementation~\atlas\ derives
bounds on the expected amortised costs. Neither is supported by \Absynth\ or \ecoimp.
Martingale based techniques have been implemented, e.g.,
by Peixin~Wang et al.~\cite{WFGCQS:PLDI:19}. Related results have been reported by
Moosbrugger et al.~\cite{MoosbruggerBKK21}.
Meyer et al.~\cite{MeyerHG21} provide an extension of the \koat\ tool,
generalising the concept of alternating size and runtime analysis
to probabilistic programs. Again, these innovative tools are not suited to the
benchmarks considered in our work.
With respect to probabilistic \emph{functional} programs, Di Wang et al.~\cite{WangKH20} provided the only prior expected cost analysis of (deterministic) probabilistic programs; this work is most closely related to our contributions. Indeed, our typing rule~\rulecoin\ stems from~\cite{WangKH20} and the soundness proof wrt.\ the
big-step semantics is conceptually similar. Nevertheless, our contributions strictly generalise their results.
First, our core language is based on a simpler semantics, giving rise to cleaner formulations of our soundness theorems. Second, our type-and-effect provides two different typing rules for ticking, a fact we can
capitalise on in additional strength of our prototype implementation.
Finally, our amortised analysis allows for logarithmic potential functions.

A bulk of research concentrates on specific forms of \emph{martingales} or \emph{Lyapunov ranking functions}.
All these works, however, are somewhat orthogonal to our contributions, as foremostly \emph{termination} (ie. AST or PAST) is studied, rather than computational complexity. Still these approaches can be partially
suited to a variety of quantitative program properties, see \cite{TOUH:ATVA:18} for an overview, but
are incomparable in strength to the results established here.

\paragraph*{Structure.}
In the next section, we provide a bird's eye view on our approach. Sections~\ref{sec:language} and~\ref{sec:semantics} detail the core probabilistic language
employed, as well as its small- and big-step semantics. In Section~\ref{sec:typesystem} we
we introduce the novel type-and-effect system formalising and state soundness of the system wrt.\ the respective semantics.
In Section~\ref{sec:implementation} we present evaluation results of our prototype implementation~\atlas.
Finally, we conclude in Section~\ref{sec:conclusion}.

%% file: overview.tex
In this section, we first sketch our approach on an introductory example and
then detail the benchmarks and results depicted in Table~\ref{tab:results} in the Introduction.

\subsection{Introductory Example}

Consider the definition of the function \descend, depicted in Figure~\ref{fig:7}.
The \emph{expected} amortised complexity of \descend\ is $\log(\size{t})$, where $\size{t}$ denotes the size of a tree (defined as the number of leaves of the tree).%
\footnote{An amortised analysis may always default to a wort-case analysis. In particular the analysis of~\descend\ in this section can be considered as a worst-case analysis. However, we use the example to illustrate the general setup of our amortised analysis.}
Our analysis is set up in terms of template potential functions with unknown coefficients, which will be instantiated by our analysis.
Following~\cite{hofmann2021typebased,LMZ:2021}, our potential functions are composed of two types of resource functions, which can express \emph{logarithmic} amortised cost: For a sequence of $n$ trees $\seq{t}$ and coefficients $a_i \in \N, b \in \Z$, with $\sum_{i=1}^n a_i + b \geqslant 0$, the resource function
$p_{(\seq{a},b)}(\seq{t}) \defsym \log(\seqx{a_\i \cdot \size{t_\i}}[n][+] + b)$
denotes the logarithm of a linear combination of the sizes of the trees.
The resource function $\rk(t)$, which is a variant of Schoenmakers' potential, cf.~\cite{Schoenmakers92,Schoenmakers93,NipkowB19}, is inductively defined as
\begin{inparaenum}[(i)]
  \item $\rk(\leaf) \defsym 1$;
  \item $\rk(\tree{l}{d}{r}) \defsym \rk(l) + \log(\size{l}) + \log(\size{r}) + \rk(r)$,
\end{inparaenum}
where $l$, $r$ are the left resp.\ right child of the tree $\tree{l}{d}{r}$, and $d$ is some data element that is ignored by the resource function.
(We note that $\rk(t)$ is not needed for the analysis of \descend\ but is needed for more involved benchmarks, e.g.\ randomised splay trees.)
With these resource functions at hand, our analysis introduces the coefficients $q_\ast$, $q_{(1,0)}$, $q_{(0,2)}$, $q'_\ast$, $q'_{(1,0)}$, $q'_{(0,2)}$ and employs the following \emph{Ansatz}:%
\footnote{For ease of presentation, we elide the underlying semantics for now and simply write ``\lstinline{descend t}'' for the resulting tree $t'$, obtained after evaluating \lstinline{descend t}.}
\begin{gather*}
  q_\ast \cdot \rk(t) + q_{(1,0)} \cdot p_{(1,0)}(t) + q_{(0,2)} \cdot p_{(0,2)}(t) \geqslant c_{\text{\descend}}(t) + {}
  \\
  {} + q'_\ast \rk(\text{\descend}\ t) + q'_{(1,0)} \cdot p_{(1,0)}(\text{\descend}\ t) +
    q'_{(0,2)} \cdot p_{(0,2)}(\text{\descend}\ t)\tpkt
\end{gather*}
Here, $c_{\text{\descend}}(t)$ denotes the expected cost of executing \descend\ on tree $t$, where the cost is given by the ticks as indicated in the source code (each tick accounts for a recursive call).
The result of our analysis will be an instantiation of the coefficients, returning $q_{(1,0)} = 1$ and zero for all other coefficients, which allows to directly read off the logarithmic bound $\log(\size{t})$ of \descend.

Our analysis is formulated as a \emph{type-and-effect system}, introducing the above \emph{template potential functions} for every subexpression of the program under analysis.
The typing rules of our system give rise to a constraint system over the unknown coefficients that capture the relationship between the potential functions of the subexpressions of the program.
Solving the constraint system then gives a valid instantiation of the potential function coefficients.
Our type-and-effect system constitutes a generalisation of the type system studied in~\cite{hofmann2021typebased,LMZ:2021} to the non-deterministic and probabilistic setting, as well as an extension of the type system introduced in~\cite{WangKH20} to sublinear bounds and non-determinism.

\begin{figure}[t]
\centering
\begin{code}
descend t = match t with
  | leaf       -> leaf
  | node l a r -> if coin 1/2 (* Denotes $\color{darkgray} p = \sfrac{1}{2}$, which is default and could be omitted. *)
    then let xl = $\tick{(\flst{\scriptsize descend l})}$ in node xl a r (* The symbol $\tiny \checkmark$ denotes a tick. *)
    else let xr = $\tick{(\flst{\scriptsize descend r})}$ in node l a xr
\end{code}
\vspace{-0.3cm}
\caption{\descend\ function}
\label{fig:7}
\vspace{-0.5cm}
\end{figure}

In the following, we survey our type-and-effect system by means of example~\descend.
A partial type derivation is given in Figure~\ref{fig:8}.
For brevity, type judgements and the type rules are presented in a simplified form.
In particular, we restrict our attention to tree types, denoted as $\TreeShort$.
This omission is inessential to the actual complexity analysis.
For the full set of rules see \aword{}.
We now discuss this type derivation step by step.

Let $e$ denote the body of the function definition of \descend, cf.~ Figure~\ref{fig:7}.
Our automated analysis infers an \emph{annotated type} by verifying that the type judgement
$\tjudge{\typed{t}{\TreeShort}}{Q}{e}{\TreeShort}{Q'}$ is derivable.
Types are decorated with \emph{annotations} $Q \defsym [q_\ast,q_{(1,0)},q_{(0,2)}]$ and~$Q' \defsym [q'_\ast,q'_{(1,0)},q'_{(0,2)}]$---employed to express the potential carried by the arguments to \descend\ and its results.
Annotations fix the coefficients of the resource functions in the corresponding potential functions, e.g.,
\begin{inparaenum}[(i)]
  \item $\potential{\typed{t}{\Tree}}{Q} \defsym q_\ast \cdot \rk(t) + q_{(1,0)} \cdot p_{(1,0)}(t) + q_{(0,2)}\cdot p_{(0,2)}(t)$ and
  \item $\potential{\typed{e}{\Tree}}{Q'} \defsym q'_\ast \cdot \rk(e) + q'_{(1,0)} \cdot p_{(1,0)}(e) + q'_{(0,2)} \cdot p_{(0,2)}(e)$.
\end{inparaenum}

By our soundness theorems (see Section~\ref{Soundness}), such a typing guarantees that the \emph{expected} amortised cost of \descend\ is bounded by the expectation (\wrt\ the distribution of values in the limit) of the difference between $\potential{\typed{t}{\Tree}}{Q}$ and $\potential{\typed{\text{\descend}\ t}{\TreeShort}}{Q'}$.
Because $e$ is a \lstinline|match| expression,
the following rule is applied (we only state a restricted rule here, the general rule can be found in \aword{}):
\begin{tcolorbox}[ams equation*]
  \infer[\rulematch]{
    \tjudge{\typed{t}{\TreeShort}}{Q}{\match\ t\
        \with \text{\lstinline{|}} \leaf\ \arrow\ \leaf\
        \text{\lstinline{|}} \flsttree{l}{a}{r}\ \arrow\ e_1}{\TreeShort}{Q'}
    }{%
      \tjudge{\varepsilon}{\varnothing}{\leaf}{\TreeShort}{Q'}
      &
      \tjudge{\typed{l}{\TreeShort},\typed{r}{\TreeShort}}{Q_1}{e_1}{\TreeShort}{Q'}
    }
\end{tcolorbox}
Here $e_1$ denotes the subexpression of $e$ that corresponds to the \lstinline|node| case of \lstinline|match|.
Apart from the annotations $Q$, $Q_1$ and $Q'$, the rule $\rulematch$ constitutes a standard type rule for pattern matching.
With regard to the annotations $Q$ and $Q_1$, \rulematch\ ensures the correct distribution of potential by inducing the constraints
\begin{align*}
q^1_1  = q^1_2 = q_\ast && q^1_{(1,1,0)}  = q_{(1,0)} && q^1_{(1,0,0)} = q^1_{(0,1,0)} = q_\ast && q^1_{(0,0,2)} = q_{(0,2)}
\tkom
\end{align*}
%
where the constraints are immediately justified by recalling the definitions of the resource functions $p_{(\seq{a},b)}(\seq{t}) \defsym \log(\seqx{a_\i \cdot \size{t_\i}}[n][+] + b)$ and $\rk(t) = \rk(l) + \log(\size{l}) + \log(\size{r}) + \rk(r)$.

\begin{figure}[t]
  \centering
  \begin{equation*}
    \infer[\rulematch]{\tjudge{\typed{t}{\TreeShort}}{Q}{\match\ t\
        \with \text{\lstinline{|}} \leaf\ \arrow\ \leaf\
        \text{\lstinline{|}} \flsttree{l}{a}{r}\ \arrow\ e_1}{\TreeShort}{Q'}}{%
	\infer[\rulew]{\tjudge{\typed{l}{\TreeShort}, \typed{r}{\TreeShort}}{Q_1}{\cif\ \coin\ \cthen\ e_2\ \celse\ e_3}{\TreeShort}{Q'}}{%
		\infer[\rulecoin]{\tjudge{\typed{l}{\TreeShort}, \typed{r}{\TreeShort}}{Q_2}{\cif\ \coin\ \cthen\ e_2\ \celse\ e_3}{\TreeShort}{Q'}}{%
            \infer[\rulelet]{\tjudge{\typed{l}{\TreeShort}, \typed{r}{\TreeShort}}{Q_3}{\vlet\ x_l = \tick{(\text{\lstinline{descend l}})} \vin\ \flsttree{x_l}{a}{r}}{\TreeShort}{Q'}}{%
              \infer[\ruletick]{\tjudge{\typed{l}{\TreeShort}}{Q_4}{\tick{(\text{\lstinline{descend l}})}}{\TreeShort}{Q_6}}{%
                  \infer[\ruleapp]{\tjudge{\typed{l}{\TreeShort}}{Q_5}{\text{\lstinline{descend l}}}{\TreeShort}{Q_6}}{%
                    \text{\lstinline{descend:}}{ \atypdcl{\TreeShort}{Q}{\TreeShort}{Q'}}
                }
              }
            &
            \tjudge{\typed{x_l}{\TreeShort},\typed{r}{\TreeShort}}{Q_7}{\flsttree{x_l}{a}{r}}{\TreeShort}{Q'}
          }
          }
        }
      }
    \end{equation*}
    \vspace{-0.6cm}
  \caption{Partial Type Derivation for Function \descend}
  \label{fig:8}
  \vspace{-0.5cm}
\end{figure}

The next rule is a structural rule, representing a \emph{weakening} step that rewrites the annotations of the variable context.
The rule \rulew\ allows a suitable adaptation of the coefficients based on the following inequality, which holds for any substitution $\sigma$ of variables by values,
$\spotential{\typed{l}{\TreeShort}, \typed{r}{\TreeShort}}{Q_1} \geqslant
\spotential{\typed{l}{\TreeShort}, \typed{r}{\TreeShort}}{Q_2}$.
\begin{tcolorbox}[ams equation*]
\infer[\rulew]{%
  \tjudge{\typed{l}{\TreeShort}, \typed{r}{\TreeShort}}{Q_1}{e_1}{\TreeShort}{Q'}
}{%
  \tjudge{\typed{l}{\TreeShort}, \typed{r}{\TreeShort}}{Q_2}{e_1}{\TreeShort}{Q'}
}
\end{tcolorbox}
In our prototype implementation this comparison is performed \emph{symbolically}.
We use Farkas' Lemma in conjunction with two facts about the logarithm to linearise this symbolic comparison, namely the monotonicity of the logarithm and the fact that $2 + \log(x) + \log(y) \leqslant 2\log(x+y)$ for all $x,y \geqslant 1$.
For example, Farkas' Lemma in conjunction with the latter fact gives rise to
\begin{align*}
  q^1_{(0,0,2)} + 2f & \geqslant q^2_{(0,0,2)} & q^1_{(1,1,0)} - 2f &\geqslant q^2_{(1,1,0)}
  \\
  q^1_{(1,0,0)} + \phantom{2}f & \geqslant q^2_{(1,0,0)} & q^1_{(0,1,0)} + \phantom{2}f & \geqslant q^2_{(0,1,0)}                                                                    \tkom
\end{align*}
for some fresh rational coefficient $f \geqslant 0$ introduced by Farkas' Lemma.
After having generated the constraint system for \descend, the solver is free to instantiate $f$ as needed.
In fact in order to discover the bound $\log(\size{t})$ for \descend, the solver will need to instantiate $f = \sfrac{1}{2}$, corresponding to the inequality $\log(\size{l}+\size{r}) \geqslant \sfrac{1}{2} \log(\size{l}) + \sfrac{1}{2} \log(\size{r}) + 1$.

\begin{figure}[t]
\centering
\begin{code}
meld h1 h2 = match h1 with
  | leaf             -> h2
  | node h1l h1x h1r -> match h2 with
    | node h2l h2x h2r -> if h1x > h2x
      then if coin
        then (node $\tick{(\flst{\scriptsize meld h2l (node h1l h1x h1r)})}$ h2x h2r)
        else (node h2l h2x $\tick{(\flst{\scriptsize meld h2r (node h1l h1x h1r)})}$)
      else (* Omitted for brevity, symmetric to the the depicted case. *)
\end{code}
\vspace{-0.3cm}
\caption{Partial \meld\ function of Randomised Meldable Heaps}
\label{fig:meld}
\vspace{-0.5cm}
\end{figure}

So far, the rules did not refer to sampling and are unchanged from their (non-probabilistic) counterpart
introduced in~\cite{hofmann2021typebased,LMZ:2021}.
The next rule, however, formalises a coin toss, biased with probability~$p$.
Our general rule $\rulecoin$ is depicted in Figure~\ref{fig:6} and is inspired by a similar rule for coin tosses that has been recently been proposed in the literature, cf.~\cite{WangKH20}.
This rule specialises as follows to our introductory example:
\begin{tcolorbox}[ams equation*]
  \infer[\rulecoin]{%
    \tjudge{\typed{l}{\TreeShort}, \typed{r}{\TreeShort}}{Q_2}{\cif\ \coin[1][2]\ \cthen\ e_2\ \celse\ e_3}{\TreeShort}{Q'}}{%
    \begin{minipage}{60ex}
      $\tjudge{\typed{l}{\TreeShort}, \typed{r}{\TreeShort}}{Q_4}{e_3}{\TreeShort}{Q'}$
      \\[1ex]
      $\tjudge{\typed{l}{\TreeShort}, \typed{r}{\TreeShort}}{Q_3}{\vlet\ x_l = \tick{(\text{\descend}\ l)} \vin\ \flsttree{x_l}{a}{r}}{\TreeShort}{Q'}$
    \end{minipage}
    }
\end{tcolorbox}
Here $e_2$ and $e_3$ respectively, denote the subexpressions of the conditional and
in addition the crucial condition $Q_2 = \sfrac{1}{2} \cdot Q_3 + \sfrac{1}{2} \cdot Q_4$ holds.
This condition, expressing that the
corresponding annotations are subject to the probability of the coin toss, gives rise to the
following constraints (among others)
%
\begin{align*}
  q^2_{(0,0,2)} & = \sfrac{1}{2} \cdot q^3_{(0,0,2)} + \sfrac{1}{2} \cdot q^4_{(0,0,2)}
                  & q^2_{(0,1,0)} & = \sfrac{1}{2} \cdot q^3_{(0,1,0)} + \sfrac{1}{2} \cdot q^4_{(0,1,0)}
  \\
  q^2_{(1,0,0)} & = \sfrac{1}{2} \cdot q^3_{(1,0,0)} + \sfrac{1}{2} \cdot q^4_{(1,0,0)}
                  \tpkt
\end{align*}
In the following, we will only consider one alternative of the coin toss and proceed as in the partial type derivation depicted in Figure~\ref{fig:7} (\ie\ we state the \lstinline|then|-branch and omit the symmetric \lstinline|else|-branch).
Thus next, we apply the rule for the \lstinline|let|~expression.
This rule is the most involved typing rule in the system proposed in~\cite{hofmann2021typebased,LMZ:2021}.
However, for our leading example it suffices to consider the following simplified variant:
\begin{tcolorbox}[ams equation*]
  \infer[\rulelet]{\tjudge{\typed{l}{\TreeShort}, \typed{r}{\TreeShort}}{Q_3}{\vlet\ x_l = \tick{(\text{\descend}~l)} \vin\ \flsttree{x_l}{a}{r}}{\TreeShort}{Q'}}{%
    \tjudge{\typed{l}{\TreeShort}}{Q_4}{\tick{(\text{\descend}~l)}}{\TreeShort}{Q_6}
    &
    \tjudge{\typed{l}{\TreeShort}}{Q_7}{\flsttree{x_l}{a}{r}}{\TreeShort}{Q'}
  }
\end{tcolorbox}

\begin{figure}[t]
\centering
\begin{code}
splay a t = match t with
  | node cl c cr -> match cl with
    | node bl b br -> match $\tick[1][2]{(\flst{\scriptsize splay a bl})}$ with (* Recursive call costs $\color{darkgray}\sfrac{1}{2}$. *)
      | node al a1 ar -> if coin
	then $\tick[1][2]{(\flst{\scriptsize node al a1 (node ar b (node br c cr))})}$ (* Rotation costs $\color{darkgray}\sfrac{1}{2}$. *)
	else       node (node (node al a1 ar) b br) c cr (* No rotation. *)
\end{code}
\vspace{-0.3cm}
\caption{Partial \splay\ function of Randomised Splay Trees (zigzig-case)}
\label{fig:splay}
\vspace{-0.5cm}
\end{figure}

Focusing on the annotations, the rule \rulelettreecf\ suitably distributes potential assigned to the variable context, embodied in the annotation $Q_3$, to the recursive call within the \lstinline|let| expression (via annotation $Q_4$) and the construction of the resulting tree (via annotation $Q_7$).
The distribution of potential is facilitated by generating constraints that can roughly be stated as two ``equalities'', that is, (i) ``$Q_3 = Q_4 + D$'' and (ii) ``$Q_7 = D + Q_6$''.
Equality (i) states that the input potential is split into some potential $Q_4$ used for typing $\tick{(\text{\descend}~l)}$ and some remainder potential $D$ (which however is not constructed explicitly and only serves as a placeholder for potential that will be passed on).
Equality (ii) states that the potential $Q_7$ used for typing $\flsttree{x_l}{a}{r}$ equals the remainder potential $D$ plus the leftover potential $Q_6$ from the typing of $\tick{(\text{\descend}~l)}$.
The \ruletick\ rule then ensures that costs are properly accounted for by generating constraints for $Q_4 = Q_5 + 1$.
Finally, the type derivation ends by the application rule, denoted as \ruleapp, that verifies that the recursive call is well-typed \wrt\ the (annotated) signature of the function $\text{\descend} \colon \atypdcl{\TreeShort}{Q}{\TreeShort}{Q'}$, \ie\ the rule enforces that $Q_5 = Q$ and $Q_6 = Q'$.
We illustrate (a subset of) the constraints induced by \rulelet, \ruletick\ and \ruleapp:
\begin{align*}
  q^3_{(1,0,0)} & = q^4_{(1,0)} & q^3_{(0,1,0)} & = q^7_{(0,1,0)} & q'_1 & = q^6_1 & q^4_{(0,2)} & = q^5_{(0,2)} + 1
  \\
  q^3_{(0,0,2)} & = q^4_{(0,2)} & q^3_2 & = q^7_2 & q'_{(1,0)} & = q^6_{(1,0)} & q^4_{(1,0)} & = q^5_{(1,0)}
  \\
  q^3_1 & = q^4_1 & q'_{(0,2)} & = q^6_{(0,2)} & q^6_1 & = q^7_1 & q^5_{(1,0)} & = q_{(1,0)}
  \tkom
\end{align*}
where
\begin{inparaenum}[(i)]
\item the constraints in the first three columns---involving the annotations $Q_3$, $Q_4$, $Q_6$ and $Q_7$---%
stem from the constraints of the rule \rulelettreecf;
\item the constraints in the last column---involving $Q_4$, $Q_5$, $Q$ and $Q'$---%
stem from the constraints of the rule \ruletick\ and \ruleapp.
\end{inparaenum}
For example, $q^3_{(1,0,0)} = q^4_{(1,0)}$ and $q^3_{(0,1,0)}  = q^7_{(0,1,0)}$ distributes the part of the logarithmic potential represented by $Q_3$ to $Q_4$ and $Q_7$;
$q^6_1 = q^7_1$ expresses that the rank of the result of evaluating the recursive call can be employed in the construction of the resulting tree $\tree{x_l}{a}{r}$;
$q^4_{(1,0)} = q^5_{(1,0)}$ and $q^4_{(0,2)}  = q^5_{(0,2)} + 1$ relate the logarithmic resp. constant potential according to the tick rule, where the addition of one accounts for the cost embodied by the tick rule;
$q^5_{(1,0)} = q_{(1,0)}$ stipulates that the potential at the recursive call site must match the function type.

Our prototype implementation \atlas\ collects all these constraints and solves them fully automatically.
Following~\cite{hofmann2021typebased,LMZ:2021}, our implementation in fact searches for a solution that minimises the resulting complexity bound.
For the \descend\ function, our implementation finds a solution that sets $q_{(1,0)}$ to $1$, and all other coefficients to zero.
Thus, the logarithmic bound $\log(\size{t})$ follows.

\subsection{Overview of Benchmarks and Results}

\paragraph{Randomised Meldable Heaps.} Gambin et al.~\cite{GambinM98} proposed meldable heaps as a simple priority-queue data structure that is guaranteed to have expected logarithmic cost for all operations.
All operations can be implemented in terms of the \meld\ function, which takes two heaps and returns a single heap as a result.
The partial source code of \meld\ is given in Figure~\ref{fig:meld} (the full source code of all examples can be found in \aword{}).
Our tool \atlas\ fully-automatically infers the bound $\log(\size{h1}) + \log(\size{h2})$ on the expected cost of \meld.

\begin{figure}[t]
\centering
\begin{code}
insert d t = match t with
  | leaf       -> node leaf d leaf
  | node l a r -> if coin 1/2     (* Assuming probability $\sfrac{1}{2}$ for $a < d$. *)
    then node $\tick{(\flst{\scriptsize insert d l})}$ a r
    else node l a $\tick{(\flst{\scriptsize insert d r})}$
\end{code}
\vspace{-0.2cm}
\caption{\bstinsert\ function of a Binary Search Tree with randomized comparison}
\label{fig:bst-insert}
\vspace{-0.5cm}
\end{figure}

\paragraph{Randomised Splay Trees. } Albers et al. in~\cite{AlbersK02} proposed these splay trees as a variation of deterministic splay trees~\cite{ST:1985}, which have better expected runtime complexity (the same computational complexity in the O-notation but with smaller constants). Related results have been obtained by Fürer~\cite{Furer99}.
The proposal is based on the observation that it is not necessary to rotate the tree in every (recursive) splaying operation but that it suffices to perform rotations with some fixed positive probability in order to reap the asymptotic benefits of self-adjusting search trees.
The theoretical analysis of randomised splay trees~\cite{AlbersK02} starts by refining the cost model of~\cite{ST:1985}, which simply counts the number of rotations, into one that accounts for recursive calls with a cost of $c$ and for rotations with a cost of $d$.
We present a snippet of a functional implementation of randomised splay trees in Figure~\ref{fig:splay}.
We note that in this code snippet we have set $c = d = \sfrac{1}{2}$; this choice is arbitrary; we have chosen these costs in order to be able to compare the resulting amortised costs to the deterministic setting of~\cite{LMZ:2021}, where the combined cost of the recursive call and rotation is set to $1$; we note that our analysis requires fixed costs $c$ and $d$ but these constants can be chosen by the user; for example one can set $c=1$ and $d=2.75$ corresponding to the costs observed during the experiments in~\cite{AlbersK02}.
Likewise the probability of the coin toss has been arbitrarily set to $p = \sfrac{1}{2}$ but could be set differently by the user.
(We remark that to the best of our knowledge no theoretical analysis has been conducted on how to chose the best value of p for given costs $c$ and $d$.)
Our analysis is able to fully automatically infer an amortised complexity bound of $\sfrac{9}{8}\log(\size{t})$ for \splay\ (with $c$, $d$ and $p$ fixed as above), which improves on the complexity bound of $\sfrac{3}{2}\log(\size{t})$ for the deterministic version of \splay\ as reported in~\cite{LMZ:2021}, confirming that randomisation indeed improves the expected runtime.

We remark on how the amortised complexity bound of $\sfrac{9}{8}\log(\size{t})$ for \splay\ is computed by our analysis.
Our tool \atlas\ computes an annotated type for \splay\ that corresponds to the inequality
$\sfrac{3}{4} \rk(t) + \sfrac{9}{8} \log(\size{t}) + \sfrac{3}{4}  \geqslant c_{\text{\splay}}(t) + \sfrac{3}{4} \rk(\text{\splay}\ t) + \sfrac{3}{4}$.
By setting $\phi(t) := \rk(t) + \sfrac{3}{4}$ as potential function in the sense of Tarjan and Sleator~\cite{ST:1985,Tarjan:1985}, the above inequality allows us to directly read out an upper bound on the amortised complexity $a_{\text{\splay}}(t)$ of \splay\ (we recall that the amortised complexity in the sense of Tarjan and Sleator is defined as the sum of the actual costs plus the output potential minus the input potential):
$a_{\text{\splay}}(t) =  c_{\text{\splay}}(t) + \phi(\text{\splay}\ t) - \phi(t) \leqslant \sfrac{9}{8} \cdot \log(\size{t})$.

\paragraph{Probabilistic Analysis of Binary Search Trees.}
We present a probabilistic analysis of a deterministic binary search tree, which offers the usual \bstcontains, \bstinsert, and \bstdelete\ operations, where \bstdelete\ uses \bstdelmax\, given in Figure~\ref{fig:bst-delmax}, as a subroutine (the source code of the missing operations is given in \aword{}).
We assume that the elements inserted, deleted and searched for are equally distributed;
hence, we conduct a probabilistic analysis by replacing every comparison with a coin toss of probability one half.
We will refer to the resulting data structure as Coin Search Tree in our benchmarks.
The source code of \bstinsert\ is given in Figure~\ref{fig:bst-insert}.
Our tool \atlas\ infers an logarithmic expected amortised cost for all operations, \ie, for \bstinsert\ and \bstdelmax\ we obtain
\begin{inparaenum}[(i)]
\item $\sfrac{3}{2} \rk(t) + \sfrac{1}{2} \log(\size{t})  \geqslant c_{\text{\bstinsert}}(t) + \sfrac{3}{2} \rk(\text{\bstinsert}\ t)$; and
\item $\sfrac{3}{2} \rk(t) + \sfrac{1}{2} \log(\size{t})  \geqslant c_{\text{\bstdelmax}}(t) + \sfrac{3}{2} \rk(\text{\bstdelmax}\ t)$,
\end{inparaenum}
from which we obtain an expected amortised cost of $\sfrac{1}{2} \log(\size{t})$ for both functions.

\begin{figure}[t]
\centering
\begin{code}
(*pre-condition: t is not a leaf*)
delete_max t = match t with
  | node l b r -> match r with
    | leaf         -> (l,b)
    | node rl c rr -> match rr with
      | leaf -> ((node l b rl),c)
      | rr   -> let (t',max) = $\tick{(\flst{\scriptsize delete\_max rr})}$ in match t' with
        | node rrl1 x xa -> (node (node (node l b rl) c rrl1) x xa,max)
\end{code}
\vspace{-0.2cm}
\caption{\bstdelmax\ function of a Coin Search Tree with one rotation}
\label{fig:bst-delmax}
\vspace{-0.5cm}
\end{figure}

%% file: language.tex
\paragraph*{Preliminaries.}

Let $\Rplus$ denote the non-negative reals and $\Rplusinfty$ their extension by $\infty$.
We are only concerned with \emph{discrete distributions} and drop ``discrete'' in the following.
Let $A$ be a countable set and let $\Dist{A}$ denote the set of \emph{(sub)distributions} $d$ over $A$, whose
support $\supp{\mu} \defsym \{a \in A \mid \mu(a) \not= 0\}$ is countable. Distributions are denoted by Greek letters.
For $\mu \in \Dist{A}$, we may write $\mu = \{a^{p_i}_i\}_{i \in I}$, assigning probabilities $p_i$ to $a_i \in A$ for every $i \in I$, where $I$ is a suitable chosen index set. We set $\prob{\mu} \defsym \sum_{i \in I} p_i$. If the support is finite, we simply write $\mu = \{a^{p_1}_1,\dots,a^{p_n}_n\}$
The \emph{expected value} of a function $f \colon A \to \Rplus$ on $\mu \in \Dist{A}$ is defined as
$\Expect{\mu}{f} \defsym \sum_{a \in \supp{\mu}} \mu(a) \cdot f(a)$. Further, we denote by
$\sum_{i \in I} p_i \cdot \mu_i$ the \emph{convex combination of distributions $\mu_i$}, where $\sum_{i\in I} p_i \leqslant 1$.
As by assumption $\sum_{i\in I} p_i \leqslant 1$, $\sum_{i \in I} p_i \cdot \mu_i$ is always a (sub-)distribution.

\begin{figure}[t]
\centering
\begin{align*}
  \circ & \Coloneqq \textup{\lstinline{<}}
	\mid \textup{\lstinline{>}}
	\mid \textup{\lstinline{=}}
  \\
  e & \Coloneqq f~x_1~\dots~x_n
   && \mid \tick[$\color{black}a$][$\color{black}b$]{e}
  \\
    & \mid \false \mid \true \mid e_1 \circ e_2
    && \mid \cif\ x\ \cthen\ e_1\ \celse\ e_2
  \\
    &
    && \mid \cif\ \nondet \ \cthen\ e_1\ \celse\ e_2
  \\
    &
    && \mid \cif\ \coin[$\color{black}a$][$\color{black}b$] \ \cthen\ e_1\ \celse\ e_2
  \\
    & \mid \leaf \mid \textup{\tree{x_1}{x_2}{x_3}}
    && \mid \match\ x\ \with\
      \textup{\lstinline{|}}\ \leaf \arrow e_1\
      \textup{\lstinline{|}}\ \textup{\tree{x_1}{x_2}{x_3}} \arrow e_2
  \\
    & \mid \textup{\pair{x_1}{x_2}}
    && \mid \match\ x\ \with\
      \textup{\lstinline{|}}\ \textup{\pair{x_1}{x_2}} \arrow e\
  \\
    & \mid \vlet\ x~\equal~e_1\ \vin\ e_2
    && \mid x
  \\
\end{align*}
\vspace{-1cm}
\caption{A Core Probabilistic (First-Order) Programming Language}
\label{fig:1}
\vspace{-0.5cm}
\end{figure}

\paragraph*{Syntax.}
In Figure~\ref{fig:1}, we detail the syntax of our core probabilistic (first-order) programming language.
With the exception of ticks, expressions are given in \lnf\ to simplify the presentation
of the operational semantics and the typing rules.
In order to ease the readability, we make use of mild
syntactic sugaring in the presentation of actual code (as we already did above).

To make the presentation more succinct, we assume only the following types: a set of
\emph{base types} $\mathcal{B}$ such as Booleans $\Bool = \{\true, \false\}$, integers $\Int$,
or rationals $\Rat$, product types, and binary trees $\Tree$,
whose internal nodes are labelled with elements $\typed{b}{\Base}$, where $\Base$ denotes
an arbitrary base type. \emph{Values} are either of base types, trees or pairs of values.
We use lower-case Greek letters (from the beginning of the alphabet) for the denotation of types.
Elements $\typed{t}{\Tree}$ are defined by the following grammar which fixes notation.
%
$t \Coloneqq \leaf \mid \tree{t_1}{b}{t_2}$.
The size of a tree is the number of leaves: $\size{\leaf} \defsym 1$,
$\size{\tree{t}{a}{u}} \defsym \size{t} + \size{u}$.

We skip the standard definition of integer constants $n \in \Z$ as well as variable
declarations, cf.~\cite{Pierce:2002}. Furthermore, we omit binary operators with the exception of
essential comparisons.
As mentioned, to represent sampling we make use of a dedicated \lstinline|if|-\lstinline|then|-\lstinline|else| expression, whose guard evaluates to $\true$ depending on a coin toss with fixed probability. Further, non-deterministic
choice is similarly rendered via an \lstinline|if|-\lstinline|then|-\lstinline|else| expression.
%
Moreover, we make use of \emph{ticking}, denoted by an operator $\tick[\ensuremath{\color{black} a}][\ensuremath{\color{black} b}]{\cdot}$ to annotate costs, where $a$, $b$ are optional and default to one.
Following Avanzini et al.~\cite{ABD:2021}, we represent ticking $\tick{\cdot}$ as an operation, rather than in~\lnf,
as in~\cite{WangKH20}. This allows us to suit a big-step semantics that only accumulates the cost of terminating expressions. The set of all expressions is denoted $\Expr$.

A \emph{typing context} is a mapping from variables $\VS$ to types.
Type contexts are denoted by upper-case Greek letters, and the empty context is denoted $\varepsilon$.
A program $\Program$ consists of a signature $\FS$ together with a set of function definitions of the form
$f~x_1~\dots~x_n = e_f$, where the $x_i$ are variables and $e_f$ an expression.
When considering some expression $e$ that includes function calls we will always assume that these function calls are defined by some program $\Program$.
A \emph{substitution} or (\emph{environment}) $\sigma$ is a mapping from variables to values that respects
types. Substitutions are denoted as sets of assignments: $\sigma = \{\seqx{x_\i\mapsto t_\i}\}$.
We write $\dom(\sigma)$ 
to denote the domain 
of $\sigma$.

\begin{figure}[t]
\begin{center}
  \begin{tabular}{l@{\qquad}l}
    $\flstk{let }x\flst{ = }w \flstk{ in }e_2 \red e_2[x \mapsto w]$
    & $\flstk{if }\true \flstk{ then }e_1\flstk{ else }e_2 \red e_1$
    \\[3pt]
    $\flstk{if } \flstp{coin~} a\flst{/}b \flstk{ then }e_1\flstk{ else }e_2 \red \{e_1^{\sfrac{a}{b}},e_2^{1-\sfrac{a}{b}}\}$
    & $\flstk{if }\false \flstk{ then }e_1\flstk{ else }e_2 \red e_2$
    \\[3pt]
    $\flstk{if } \flstp{nondet} \flstk{ then }e_1\flstk{ else }e_2 \red e_1$
    & $\flstk{if } \flstp{nondet} \flstk{ then }e_1\flstk{ else }e_2 \red e_2$
    \\[3pt]
    $\flstk{match}\ \leaf\ \flstk{with}
    \begin{array}[t]{l}%
    \flst{|}\flstc{leaf} \flst{->} e_1 \\
    \flst{|}\text{\flsttree{$x_0$}{$x_1$}{$x_2$}} \flst{->} e_2
    \end{array}
    \red e_1$
    &  $\m{f}~x_1\sigma~\dots~x_k\sigma \red  e\sigma$
    \\[3pt]
    \multicolumn{2}{l}{%
    $\flstk{match}\ \text{\flsttree{$t$}{$a$}{$u$}}\ \flstk{with}
    \begin{array}[t]{l}%
    \flst{|}\flstc{leaf} \flst{->} e_1 \\
    \flst{|}\text{\flsttree{$x_0$}{$x_1$}{$x_2$}} \flst{->} e_2
    \end{array}
    \red e_2$%
    }
    \\[3pt]
    $\flstk{match}\ \text{\flstpair{$t$}{$u$}} \
      \flstk{with}\ \flst{|} \text{\flstpair{$t$}{$u$}} \flst{->} e
      \red e$
    & $\tick[$a$][$b$]{e} \toop{\sfrac{a}{b}} e$
    \\[1pt]
  \end{tabular}
\end{center}

Here we assume $\m{f}~x_1~\dots~x_k \flst{ = } e \in \Program$, $\sigma$ a substitution respecting the signature of $f$ and $w$ is a value.
\vspace{-0.2cm}
\caption{One-Step Reduction Rules}
\label{fig:3}
\vspace{-0.5cm}
\end{figure}

%% file: semantics.tex

\paragraph*{Small-Step Semantics.}
%
The small-step semantics is formalised as a (weighted) non-deterministic, probabilistic abstract reduction system~\cite{BG:RTA:05,ADY:2020} over $\MDist{\Expr}$. In this way (expected) cost, non-determinism and probabilistic sampling are taken care of.
Informally, a probabilistic abstract reduction system is a transition systems where reducts are chosen from a probability distribution.
A reduction \wrt\ such a system is then given by a stochastic process~\cite{BG:RTA:05}, or equivalently,
as a reduction relation over \emph{multidistributions}~\cite{ADY:2020}, which arise naturally in the
context of non-determinism (we refer the reader to~\cite{ADY:2020} for an example that illustrates the advantage of multidistributions in the presence of non-determinism).
More precisely, multidistributions are countable \emph{multisets} $\prms{a_i^{p_i} }_{i \in I}$ over pairs $p_i \colon a_i$
of \emph{probabilities} $0 < p_i \leqslant 1$ and \emph{objects} $a_i \in A$ with $\sum_{i \in I} p_i \leqslant 1$. (For ease of presentation, we do not distinguish notationally between sets and multisets.)
Multidistributions over objects $A$ are denoted by $\MDist{A}$. For a multidistribution $\mu \in \MDist{A}$ the induced
distribution $\overline{\mu} \in \Dist{A}$ is defined in the obvious way by summing up the probabilities of equal
objects.

Following~\cite{AMS:2020}, we equip transitions with (positive) weights, amounting to the cost of the transition.
Formally, a \emph{(weighted) Probabilistic Abstract Reduction System} (PARS) on  a countable set $A$ is a ternary relation
$\cdot \toop{\cdot} \cdot \subseteq A \times \Rplus \times \Dist{A}$. For $a \in A$, a rule $a \toop{c} \{ b^{\mu(b)} \}_{b \in A}$ indicates that $a$ reduces to $b$ with probability $\mu(b)$ and cost $c \in \Rplus$. Note that any right-hand-side of a PARS is supposed to be a \emph{full} distribution, \ie\
the probabilities in $\mu$ sum up to $1$. Given two objects $a$ and $b$,
$a \toop{c} \{ b^1 \}$ will be written $a \toop{c} b$ for brevity. An object $a \in A$ is called \emph{terminal} if there is no rule $a \toop{c} \mu$,
denoted $a \not\toop{}$.

We suit the one-step reduction relation $\red$ given in Figure~\ref{fig:3} as a (non-deterministic) PARS over multidistributions. As above, we sometimes identify Dirac distributions $\{e^1\}$ with $e$. \emph{Evaluation contexts} are formed by \lstinline|let| expressions, as in the following grammar:
$\context \Coloneqq \hole \mid \flstk{let }x\flst{ = } \context\ \flstk{ in }e$.
We denote with $\context[e]$ the result of substitution the empty context $\hole$ with expression~$e$.
Contexts are exploited to lift the one-step reduction to a ternary weighted reduction relation ${\tomulti{\cdot}} \subseteq {\MDist{\Expr} \times \Rplusinfty \times \MDist{\Expr}}$, cf.~Figure~\ref{fig:2}. (In (Conv), $\biguplus$ refers to the usual notion of multiset union.)

The relation $\tomulti{\cdot}$ constitutes the operational (small-step) semantics of our simple probabilistic function language. Thus $\mu \tomulti{c} \nu$ states that the submultidistribution of objects $\mu$ evolves to a submultidistribution of reducts $\nu$ in one step, with an expected cost of $c$.
Note that since $\toop{\cdot}$ is non-deterministic, so is the reduction relation $\tomulti{\cdot}$.
%
We now define the evaluation of an expression $e \in \Expr$ wrt. to the small-step  relation $\tomulti{\cdot}$:
We set $e \tomulti{c}_\infty \mu$, if there is a (possibly infinite) sequence $\prms{e^1} \tomulti{c_1} \mu_1 \tomulti{c_2} \mu_2 \tomulti{c_3} \dots$ with $c = \sum_{n \geqslant} c_n$ and $\mu = \lim_{n \to \infty} \rst{\overline{\mu_n}}$, where $\rst{\overline{\mu_n}}$ denotes the restriction of the distribution $\overline{\mu_n}$ (induced
by the multidistribution $\mu_n$) to a (sub-)distribution over values.
Note that the $\rst{\overline{\mu_n}}$ form a CPO \wrt\ the pointwise ordering, cf.~\cite{Winskel:93}.
Hence, the fixed point $\mu =\lim_{n \to \infty} \rst{\overline{\mu_n}}$ exists.
We also write $e \tomulti{}_\infty \mu$ in case the cost of the evaluation is not important.

\paragraph*{(Positive) Almost Sure Termination.}
A program $\Program$ is \emph{almost surely terminating} (\emph{AST}) if for any substitution $\sigma$, and any evaluation $e\sigma \tomulti{}_\infty \mu$, we have that $\mu$ forms a full distribution.
For the definition of positive almost sure termination we assume that every statement of $\Program$ is enclosed in an ticking operation with cost one; we note that such a cost models the length of the computation.
We say $\Program$ is \emph{positively almost surely terminating} (\emph{PAST}), if for any substitution $\sigma$, and any evaluation $e\sigma \tomulti{c}_\infty \mu$, we have $c < \infty$.
It is well known that PAST implies AST, cf.~\cite{BG:RTA:05}.

\begin{figure}[t]
  \centering
     \begin{prooftree}
      \hypo{v \not\toop{} \phantom{\tomulti{c }}}
      \infer1[(NF)]{\prms{v^1} \tomulti{0} \prms{v^1}}
    \end{prooftree}
    \quad
    \begin{prooftree}
      \hypo{e \toop{c} \{ e_i^{p_i} \}_{i \in I}}
      \infer1[(Step)]{ \prms{\context[e^1]} \tomulti{c } \prms{\context[e_i]^{p_i}}_{i \in I} }
    \end{prooftree}
    \quad
    \begin{prooftree}
      \hypo{\mu_i \tomulti{c_i} \nu_i}
      \hypo{\sum_i p_i \leqslant 1}
      \infer2[(Conv)]{\biguplus_{i} p_i \cdot \mu_i \tomulti{\sum_i p_i c_i} \biguplus_{i} p_i \cdot \nu_i}
    \end{prooftree}
\vspace{-0.2cm}
\caption{Probabilistic Reduction Rules of Distributions of Expressions}
\label{fig:2}
\vspace{-0.5cm}
\end{figure}

\paragraph*{Big-Step Semantics.}
We now define the aforementioned big-step semantics.
We first define approximate judgments $\evaln{\sigma}{n}{c}{e}{\mu}$, see Figure~\ref{fig:big-step}, which say that in derivation trees with depth up to $n$ the expression $e$ evaluates to a subdistribution $\mu$ over values with cost $c$.
We now consider the cost $c_n$ and subdistribution $\mu_n$ in $\evaln{\sigma}{n}{c_n}{e}{\mu_n}$ for $n \rightarrow \infty$. Note that the subdistributions $\mu_n$ in $\evaln{\sigma}{n}{c_n}{e}{\mu_n}$ form a CPO \wrt\ the pointwise ordering, cf.~\cite{Winskel:93}.
Hence, there exists a fixed point $\mu = \lim_{n \rightarrow \infty} \mu_n$.
Moreover, we set $c = \lim_{n \rightarrow \infty} c_n$ (note that either $c_n$ converges to some real $c \in \Rplusinfty$ or we have $c = \infty$).
We now define the big-step judgments $\eval{\sigma}{c}{e}{\mu}$ by setting $\mu = \lim_{n \rightarrow \infty} \mu_n$ and $c = \lim_{n \rightarrow \infty} c_n$ for $\evaln{\sigma}{n}{c_n}{e}{\mu_n}$. We want to emphasise that the cost $c$ in $\eval{\sigma}{c}{e}{\mu}$ only counts the ticks on terminating computations.

\begin{figure}[t!]
\begin{mathpar}
\input{rules/sem/nonval.tex}
\and
\input{rules/sem/leaf.tex}
\and
\input{rules/sem/node.tex}
\and
\input{rules/sem/var.tex}
\and
\input{rules/sem/pair.tex}
\and
\input{rules/sem/app.tex}
\and
\input{rules/sem/let.tex}
\and
\input{rules/sem/matchleaf.tex}
\and
\input{rules/sem/tickB.tex}
\and
\input{rules/sem/matchnode.tex}
\and
\input{rules/sem/itefalse.tex}
\and
\input{rules/sem/matchpair.tex}
\and
\input{rules/sem/itetrue.tex}
\and
\input{rules/sem/nondetA.tex}
\and
\input{rules/sem/nondetB.tex}
\and
\input{rules/sem/coin.tex}
\end{mathpar}

Here $\sigma[x \mapsto w]$ denotes the update of the environment $\sigma$ such that $\sigma[x \mapsto w](x) = w$ and the value of all other variables remains unchanged.
For function application we set $\sigma' \defsym \{\seqx{y_\i \mapsto x_\i \sigma}[k]\}$.
In the rules covering $\match$ we set $\sigma'' \defsym \sigma \dunion \{x_0 \mapsto t, x_1 \mapsto a, x_2 \mapsto u\}$ and $\sigma''' \defsym \sigma \dunion \{x_0 \mapsto t, x_2 \mapsto u\}$ for trees and tuples respectively.
\vspace{-0.2cm}
\caption{Big-Step Semantics.}
\label{fig:big-step}
\vspace{-0.5cm}
\end{figure}

\begin{theorem}[Equivalence]
  \label{t:3}
  Let $\Program$ be a program and $\sigma$ a substitution.
  Then,
  \begin{inparaenum}[(i)]
  \item $\eval{\sigma}{c}{e}{\mu}$ implies that
  $e\sigma \tomulti{c'}_\infty \mu$ for some $c' \ge c$, and
  \item $e\sigma \tomulti{c}_\infty \mu$ implies that $\eval{\sigma}{c'}{e}{\mu}$ for some $c' \leqslant c$.
  \end{inparaenum}
  Moreover, if $e\sigma$ almost-surely terminates, we can choose $c=c'$ in both cases.
\end{theorem}


%% file: rules/sem/nonval.tex
\infer
  {\evaln{\sigma}{0}{0}{e}{\{\}}}
  {e \text{ is not a value}}

%% file: rules/sem/leaf.tex
\infer{\evaln{\sigma}{0}{0}{\flstc{leaf}}{\{\flstc{leaf}^1\}}}{}

%% file: rules/sem/node.tex
\infer{\evaln{\sigma}{0}{0}{\text{\flsttree{$x_1$}{$x_2$}{$x_3$}}}{\{(\text{\flsttree{$t$}{$b$}{$u$}})^1\}}}%
  {x_1\sigma = t
  &x_2\sigma = b
  &x_3\sigma = u} 

%% file: rules/sem/var.tex
\infer{\evaln{\sigma}{0}{0}{x}{\{v^1\}}}{x\sigma = v}

%% file: rules/sem/pair.tex
\infer{\evaln{\sigma}{0}{0}{\text{\flstpair{$x_1$}{$x_2$}}}{\{{\text{\flstpair{$t$}{$u$}}}^1\}}}%
{x_1\sigma = t
  &
  x_2\sigma = u
}

%% file: rules/sem/app.tex
\infer
  {\evaln{\sigma}{n+1}{c}{f~x_1~\dots~x_k}{\mu}}
  {f~y_1~\dots~y_k \flst{ = } e \in \Program & \evaln{\sigma'}{n}{c}{e}{\mu}}

%% file: rules/sem/let.tex
\infer{\evaln{\sigma}{n+1}{c_1 + \sum_{w \in \supp{\nu}} \nu(w) \cdot c_w}{\flstk{let }x\flst{ = }e_1\flstk{ in }e_2}{\sum_{w \in \supp{\nu}} \nu(w) \cdot \mu_w}}%
  {%
  \evaln{\sigma}{n}{c_1}{e_1}{\nu}
  &
  \text{for all $w \in \supp{\nu}$:
  $\evaln{\sigma[x \mapsto w]}{n}{c_w}{e_2}{\mu_w}$
  }
  }

%% file: rules/sem/matchleaf.tex
\infer
  {\evaln{\sigma}{n+1}{c}{\flstk{match }x\flstk{ with} \begin{array}[t]{l}%
    \flst{| } \flstc{leaf} \flst{ -> } e_1 \\
    \flst{| } \text{\flsttree{$x_0$}{$x_1$}{$x_2$}} \flst{ -> } e_2
  \end{array}}{\mu}}
  {x\sigma = \flstc{leaf} & \evaln{\sigma}{n}{c}{e_1}{\mu}}

%% file: rules/sem/matchnode.tex
\infer
  {\evaln{\sigma}{n+1}{c}{\flstk{match }x\flstk{ with} \begin{array}[t]{l}%
    \flst{| } \flstc{leaf} \flst{ -> }e_1 \\
    \flst{| } \text{\flsttree{$x_0$}{$x_1$}{$x_2$}} \flst{ -> } e_2
  \end{array}}{\mu}}
  {x\sigma = \text{\flsttree{$t$}{$a$}{$u$}} & \evaln{\sigma''}{n}{c}{e_2}{\mu}}

%% file: rules/sem/itefalse.tex
\infer
  {\evaln{\sigma}{n+1}{c}{\flstk{if }x\flstk{ then }e_1\flstk{ else }e_2}{\mu}}
  {x\sigma = \flstk{false} & \evaln{\sigma}{n}{c}{e_2}{\mu}}

%% file: rules/sem/matchpair.tex
\infer
  {\evaln{\sigma}{n+1}{c}{\flstk{match }x\flstk{ with} \begin{array}[t]{l}%
    \flst{| } \text{\flstpair{$x_1$}{$x_2$}} \flst{ -> } e \\
  \end{array}}{\mu}}
  {x\sigma = \text{\flstpair{$t$}{$u$}} & \evaln{\sigma'''}{n}{c}{e}{\mu}}

%% file: rules/sem/itetrue.tex
\infer
  {\evaln{\sigma}{n+1}{c}{\flstk{if }x\flstk{ then }e_1\flstk{ else }e_2}{\mu}}
  {x\sigma=\flstk{true} & \evaln{\sigma}{n}{c}{e_1}{\mu}}

%% file: rules/sem/nondetA.tex
\infer
  {\evaln{\sigma}{n+1}{c}{\flstk{if } \flstp{nondet} \ \flstk{ then }e_1\flstk{ else }e_2}{\mu}}
  {\evaln{\sigma}{n}{c}{e_1}{\mu}}

%% file: rules/sem/nondetB.tex
\infer
  {\evaln{\sigma}{n+1}{c}{\flstk{if } \flstp{nondet} \ \flstk{ then }e_1\flstk{ else }e_2}{\mu}}
  {\evaln{\sigma}{n}{c}{e_2}{\mu}}

%% file: rules/sem/coin.tex
\infer
  {\evaln{\sigma}{n+1}{pc_1 + (1-p)c_2}{\flstk{if } \coin[$a$][$b$] \flstk{ then }e_1\flstk{ else }e_2}{p\mu_1+(1-p)\mu_2}}
  {\evaln{\sigma}{n}{c_1}{e_1}{\mu_1} & \evaln{\sigma}{n}{c_2}{e_2}{\mu_2} & p = \sfrac{a}{b}}

%% file: typesystem.tex
\begin{figure}[t]
  \centering
  \begin{mathpar}
    \input{rules/typ/ticking.tex}
    \and
    \input{rules/typ/tickingast.tex}
  \end{mathpar}
  \vspace{-0.6cm}
  \caption{Ticking Operator. Note that $a$, $b$ are not variables but literal numbers.}
  \label{fig:4}
  \vspace{-0.4cm}
\end{figure}

\subsection{Resource Functions}

In Section~\ref{sec:overview}, we introduced a variant of Schoenmakers' potential
function, denoted as $\rk(t)$, and the additional potential
functions $p_{(\seq{a},b)}(\seq{t}) = \log(\seqx{a_\i \cdot \size{t_\i}}[n][+] + b)$, denoting the $\log$ of a linear combination of tree sizes.
We demand $\sum_{i=1}^n a_i + b \geqslant 0$ ($a_i \in \N, b \in \Z$) for well-definedness of the latter; $\log$ denotes the logarithm to the base $2$.
Throughout the paper we stipulate $\log(0) \defsym 0$ in order to avoid case distinctions.
Note that the constant function $1$ is representable: $1 = \lambda t. \log(0 \cdot \size{t}+2) = p_{(0, 2)}$.
We are now ready to state the resource annotation of a sequence of trees.

\begin{definition}
\label{d:potential}
A \emph{resource annotation} or simply \emph{annotation} of length $m$ is a sequence $Q = [q_1,\dots,q_m] \cup
[(q_{(\seq{a}[m],b)})_{a_i, b \in \N}]$, vanishing almost everywhere.
The length of $Q$ is denoted $|Q|$.
The empty annotation, that is, the annotation where all coefficients are set to zero, is denoted as $\varnothing$.
Let $\seq{t}[m]$ be a sequence of trees.
Then, the potential of $\seq[m]{t}$ \wrt\ $Q$ is given by
\begin{equation*}
\potential{\seq{t}[m]}{Q} \defsym \sum_{i=1}^m q_i \cdot \rk(t_i) + \sum_{\seq{a}[m] \in \N,b \in \Z}
q_{(\seq{a}[m],b)} \cdot p_{(\seq{a}[m],b)}(\seq{t}[m])
  \tpkt
\end{equation*}
\end{definition}

In case of an annotation of length $1$, we sometimes write $q_\ast$ instead of $q_1$.
We may also write $\potential{\typed{v}{\TypeA}}{Q}$ for the potential of a value
of type $\TypeA$ annotated with $Q$. Both notations were already used above.
Note that only values of tree type are assigned a potential.
We use the convention that the sequence elements of resource annotations
are denoted by the lower-case letter of the annotation,
potentially with corresponding sub- or superscripts.

\begin{example}
Let $t$ be a tree. To model its potential as $\log(\size{t})$
in according to Definition \ref{d:potential}, we simply set
$q_{(1,0)} \defsym 1$ and thus obtain $\potential{t}{Q} = \log(\size{t})$,
which describes the potential associated to the input tree $t$ of our leading example
\descend\ above.
\qed
\end{example}

Let $\sigma$ be a substitution, let $\Gamma$ denote a typing context
and let $\seqx{\typed{x_\i}{\TreeShort}}$ denote all tree types in $\Gamma$. A \emph{resource annotation for $\Gamma$} or simply \emph{annotation}
is an annotation for the sequence of trees $\seqx{x_\i\sigma}$.
We define the \emph{potential} of the annotated context $\Gamma {\mid} Q$ \wrt\ a substitution $\sigma$ as
$\spotential{\Gamma}{Q} \defsym \potential{\seqx{x_\i\sigma}}{Q}$.
\begin{definition}
An \emph{annotated signature} $\FS$
maps functions $f$ to sets of pairs of annotated types
for the arguments and the annotated type of the result:
\begin{equation*}
\FS(f) \defsym \bigl\{ \atypdcl{\alphaprod}{Q}{\betaprod}{Q'} \bigm|
  m = |Q|, 1 = |Q'| \bigr\}
\tpkt
\end{equation*}
We suppose $f$ takes $n$ arguments of which $m$ are trees; $m \leqslant n$ by definition.
Similarly, the return type may be the product $\betaprod$. In this case, we demand that at most
one $\beta_i$ is a tree type.%
\footnote{The restriction to at most one tree type in the resulting type is non-essential and could be lifted. However, as our benchmark functions do not require this extension, we have elided it for ease of presentation.}
\end{definition}

Instead of $\atypdcl{\alphaprod}{Q}{\betaprod}{Q'} \in \FS(f)$, we sometimes succinctly
write $\typed{f}{\atypdcl{\alpha}{Q}{\beta}{Q'}}$ where $\alpha$, $\beta$ denote
the product types $\alphaprod$, $\betaprod$, respectively. It is tacitly understood that
the above syntactic restrictions on the length of the annotations $Q$, $Q'$ are fulfilled.
For every function $f$, we also consider its \emph{cost-free} variant from which all ticks have been removed.
We collect the cost-free signatures of all functions in the set $\FScf$.

\begin{example}
  Consider the function \descend\ depicted in Figure~\ref{fig:8}. Its signature is formally represented
  as $\atypdcl{\TreeShort}{Q}{\TreeShort}{Q'}$, where $Q \defsym [q_\ast] \cup [(q_{(a,b)})_{a,b \in \Z}]$ and
  $Q' \defsym [q'_\ast] \cup [(q'_{(a,b)})_{a,b \in \Z}]$. We leave it to the reader to specify the coefficients in $Q$, $Q'$ so that the rule~\ruleapp\ as depicted in Section~\ref{sec:overview} can indeed be employed to type
  the recursive call of \descend.
\end{example}

Let $Q = [q_\ast] \cup [(q_{(a,b)})_{a,b \in \N}]$ be an annotation and let $K$
be a rational such that $q_{(0,2)}+ K \geqslant 0$.
Then, $Q' \defsym Q + K$ is defined as follows:
$Q' = [q_\ast] \cup [(q'_{(a,b)})_{a,b \in \N}]$,
where $q'_{(0,2)} \defsym q_{(0,2)} +K$ and for all $(a,b) \not= (0,2)$
$q_{(a,b)}' \defsym q_{(a,b)}$.
Recall that $q_{(0,2)}$ is the coefficient of function $p_{(0,2)}(t) = \log(0\size{t} + 2) = 1$, so
the annotation $Q+K$ increments or decrements cost from the potential induced by $Q$ by $\abs{K}$,
respectively.
Further, we define the multiplication of an annotation $Q$ by a constant $K$, denoted as $K \cdot Q$ pointwise. Moreover, let $P = [p_\ast] \cup [(p_{(a,b)})_{a,b \in \N}]$ be another annotation. Then
the addition $P+Q$ of annotations $P, Q$ is similarly defined pointwise.

\begin{figure}[t]
  \centering
  \begin{mathpar}
    \input{rules/typ/itecoin.tex}
  \end{mathpar}
  \vspace{-0.6cm}
  \caption{Conditional expression that models tossing a coin.}
  \label{fig:6}
  \vspace{-0.5cm}
\end{figure}

\subsection{Typing Rules}

The non-probabilistic part of the type system is given in \aref{fig:5,fig:5b}.
In contrast to the type system employed in~\cite{hofmann2021typebased,LMZ:2021},
the cost model is not fixed but controlled by the ticking operator.
Hence, the corresponding application rule $\ruleapp$ has been adapted.
Costing of evaluation is now handled by a dedicated \emph{ticking} operator, cf.~Figure~\ref{fig:4}.
In Figure~\ref{fig:6}, we give the rule \rulecoin\ responsible for typing probabilistic conditionals.

\begin{figure}[t]
\centering
\begin{code}
foo t = match t with
  | leaf       -> leaf
  | node l a r -> let l' = $\tick{(\flst{\scriptsize foo l})}$ in let r' = $\tick{(\flst{\scriptsize foo r})}$ in
    if nondet then l' else r'
\end{code}
\vspace{-0.3cm}
\caption{Function \foo\ illustrates the difference between $\ruletick$ and $\ruletickast$.}
\label{fig:foo}
\vspace{-0.5cm}
\end{figure}

We remark that the core type system, that is, the type system given by~\cref{fig:6} together with the remaining
rules~\aref{fig:5,fig:5b}, ignoring annotations, enjoys subject reduction and progress in the following sense,
which is straightforward to verify.

\begin{lemma}
Let $e$ be such that $\typed{e}{\alpha}$ holds. Then:
\begin{inparaenum}[(i)]
\item If $e \toop{c} \{ e^{p_i}_i \}_{i \in I}$, then $\typed{e_i}{\alpha}$ holds for all
  $i \in I$.
\item The expression $e$ is in normal form \wrt\ $\toop{c}$ iff $e$ is a value.
\end{inparaenum}
\end{lemma}

\subsection{Soundness Theorems}
\label{Soundness}

A program $\Program$ is called \emph{well-typed} if for any definition
$f(\seq{x}[n]) = e \in \Program$ and any annotated signature
$f\colon\atypdcl{\seq{\alpha}[n][\times]}{Q}{\beta}{Q'}$, we have
a corresponding typing
$\tjudge{\typed{x_1}{\alpha_1},\dots,\typed{x_k}{\alpha_k}}{Q}{e}{\beta}{Q'}$.
A program $\Program$ is called \emph{cost-free} well-typed, if the
cost-free typing relation
(denoted as $\tjudge{\cdot}{\cdot}{\cdot}{\cdot}{\cdot}$)
is used, which employs the cost-free signatures of all functions.

\begin{restatable}[Soundness Theorem for $\ruletick$]{theorem}{restatethmticknowsound} \label{t:1}
Let $\Program$ be well-typed. Suppose $\tjudge{\Gamma}{Q}{e}{\alpha}{Q'}$ and $e\sigma \tomulti{c}_\infty \mu$.
Then $\spotential{\Gamma}{Q} \geqslant c + \Expect{\mu}{\lambda v. \potential{v}{Q'}}$.
Further, if $\tjudgecf{\Gamma}{Q}{e}{\alpha}{Q'}$, then $\spotential{\Gamma}{Q} \geqslant \Expect{\mu}{\lambda v. \potential{v}{Q'}}$.
\end{restatable}

\begin{corollary}
  \label{c:1}
  Let $\Program$ be a well-typed program such that ticking accounts for all evaluation steps. Suppose $\tjudge{\Gamma}{Q}{e}{\alpha}{Q'}$.
  Then $e$ is positive almost surely terminating (and thus in particular almost surely terminating).
\end{corollary}

\begin{restatable}[Soundness Theorem for $\ruletickast$]{theorem}{restatethmtickdefersound} \label{t:2}
Let $\Program$ be well-typed.
Suppose $\tjudge{\Gamma}{Q}{e}{\alpha}{Q'}$ and $\eval{\sigma}{c}{e}{\mu}$.
Then, we have $\spotential{\Gamma}{Q} \geqslant c + \Expect{\mu}{\lambda v. \potential{v}{Q'}}$.
Further, if $\tjudgecf{\Gamma}{Q}{e}{\alpha}{Q'}$, then $\spotential{\Gamma}{Q} \geqslant \Expect{\mu}{\lambda v. \potential{v}{Q'}}$.
\end{restatable}

We comment on the trade-offs between Theorems~\ref{t:1} and~\ref{t:2}.
As stated in Corollary~\ref{c:1} the benefit of Theorem~\ref{t:1} is that when every recursive call is accounted for by a tick, then a type derivation implies the termination of the program under analysis.
The same does not hold for Theorem~\ref{t:2}.
However, Theorem~\ref{t:2} allows to type more programs than Theorem~\ref{t:1}, which is due to the fact that \ruletickast\ rule is more permissive than \ruletick. This proves very useful, in case termination is not required (or can be established by other means).

We exemplify this difference on the \foo\ function, see Figure~\ref{fig:foo}.
Theorem~\ref{t:2} supports the derivation of the type $\rk(t) + \log(\size{t}) + 1 \geqslant \rk(\text{\foo\ } t) + 1$, while Theorem~\ref{t:1} does not.
This is due to the fact that potential can be \enquote{borrowed} with Theorem~\ref{t:2}.
To wit, from the potential $\rk(t) + \log(\size{t}) + 1$ for \foo\ one can derive the potential $\rk(l') + \rk(r')$ for the intermediate context after both let-expression (note there is no +1 in this context, because the +1 has been used to pay for the ticks around the recursive calls). Afterwards one can restore the +1 by weakening $\rk(l') + \rk(r')$ to $\rk(\text{\foo\ } t) + 1$ (using in addition that $\rk(t) \geqslant 1$ for all trees $t$).
On the other hand, we cannot \enquote{borrow}
with Theorem~\ref{t:1} because the rule $\ruletick$ forces to pay the +1 for the recursive call immediately
(but there is not enough potential to pay for this).
In the same way, the application of rule \ruletickast\ and Theorem~\ref{t:2} is essential to establish the logarithmic amortised costs of randomised splay trees.
(We note that the termination of \foo\ as well as of \splay\ is easy to establish by other means: it suffices to observe that recursive calls are on sub-trees of the input tree).

%% file: implementation.tex
\begin{table}[t]
\setlength{\tabcolsep}{0.5em}
\setlength\dashlinedash{0.5pt}
\setlength\dashlinegap{1.0pt}
\centering
\begin{tabular}{|c||c:l|c:l|c:l|}
\hline
\diagbox{$p$}{$c$} & \multicolumn{2}{c}{$\sfrac{1}{2}$} & \multicolumn{2}{c}{$\sfrac{1}{3}$} & \multicolumn{2}{c|}{$\sfrac{2}{3}$} \\
\hline
\hline
$\sfrac{1}{2}$ & $\sfrac{9}{8}$   & 1.125          & $1$              &  1                  & $\sfrac{5}{4}$   & 1.25 \\
$\sfrac{1}{3}$ & $1$              & 1              & $\sfrac{5}{6}$   & $0.8\dot{3}$        & $\sfrac{7}{6}$ & $1.\dot{6}$ \\
$\sfrac{2}{3}$ & $\sfrac{55}{36}$ & $1.52\dot{7}$  & $\sfrac{77}{54}$ & $1.4\overline{259}$ & $\sfrac{44}{27}$ & $1.\overline{629}$ \\
\hline
\end{tabular}
\vspace{2mm}
\caption{Coefficients $q$ such $q \cdot \log(|t|)$ is a bound on the expected amortized complexity of \splay\ depending on the probability $p$ of a rotation and the cost $c$ of a recursive call, where the cost of a rotation is $1 - c$. Coefficients are additionally presented in decimal representation to ease comparison.}
\vspace{-4mm}
\label{tab:matrix}
\end{table}

\paragraph{Implementation.} Our prototype \atlas\ is an extension of the tool described in~\cite{LMZ:2021}.
In particular, we rely on the preprocessing steps and the implementation of the weakening rule as reported in~\cite{LMZ:2021} (which makes use of Farkas' Lemma in conjunction with selected mathematical facts about the logarithm).
We only use the fully-automated mode reported in~\cite{LMZ:2021}.
We have adapted the generation of the constraint system to the rules presented in this paper.
We rely on Z3~\cite{conf/tacas/MouraB08} for solving the generated constraints.
We use the optimisation heuristics of~\cite{LMZ:2021} for steering the solver towards solutions that minimize the resulting expected amortised complexity of the function under analysis.

\paragraph{Evaluation.}
We present results for the benchmarks described in Section~\ref{sec:overview} (plus a randomised version of splay heaps, the source code can be found in \aword{}) in Table \ref{tab:results}.
Table~\ref{tab:resources} details the computation time of our evaluations.
To the best of our knowledge this is the first time that an expected amortised cost could be inferred for these data structures.
By comparing the costs of the operations of randomised splay trees and heaps to the costs of their deterministic versions (see Table \ref{tab:results}), one can see the randomised variants have equal or lower complexity in all cases (as noted in Section~\ref{sec:overview} we have set the costs of the recursive call and the rotation to $\sfrac{1}{2}$, such that in the deterministic case, which corresponds to a coin toss with $p=1$, these costs will always add up to one).
Clearly, setting the costs of the recursion to the same value as the cost of the rotation does not need to reflect the relation of the actual costs.
A more accurate estimation of the relation of these two costs will likely require careful experimentation with data structure implementations, which we consider orthogonal to our work.
Instead, we report that our analysis is readily adapted to different costs and different coin toss probabilities.
We present an evaluation for different values of $p$, recursion cost $c$ and rotation cost $1-c$ in Table~\ref{tab:matrix}.
The memory usage according to Z3's ``max memory'' statistic was 7129MiB per instance.
The total runtime was 1H45M, with an average of 11M39S and a median of 2M33S. Two instances took longer time (36M and 49M).

\begin{table}[t]
\setlength{\tabcolsep}{0.5em}
\centering
\begin{tabular}{|l|rrrrr|}
\hline
Module & Functions & Lines & Assertions & Time       & Memory \\
\hline
\lstinline|RandSplayTree|    & 4 & 129 & 195 339 & 33M27S & 19424.44 \\
\lstinline|RandSplayHeap|    & 2 &  34 &  77 680 &  6M15S & 14914.51 \\
\lstinline|RandMeldableHeap| & 3 &  15 &  25 526 &    20S &  4290.67 \\
\lstinline|CoinSearchTree|   & 3 &  24 &  14 045 &     4S &  1798.59 \\
\lstinline|Tree|             & 1 &   5 &     151 &  $<$1S &    45.23 \\
\hline
\end{tabular}
\vspace{2mm}
\caption{Number of assertions, solving time and maximum memory usage (in mebibytes) for the combined analysis of functions per-module. The number of functions and lines of code is given for comparison.}
\vspace{-4mm}
\label{tab:resources}
\end{table}

\paragraph{Deterministic benchmarks.}
For comparison we have also evaluated our tool \atlas\ on the benchmarks of ~\cite{LMZ:2021}.
All results could be reproduced by our implementation.
In fact, for the function \flst{SplayHeap.insert} it yields an improvement of
$\sfrac{1}{4} \log(|h|)$, \ie\
$\sfrac{1}{2} \log(|h|) + \log(|h|+1) + \sfrac{3}{2}$ compared to
$\sfrac{3}{4} \log(|h|) + \log(|h|+1) + \sfrac{3}{2}$.
We note that we are able to report better results because we have generalised the resource functions $p_{(\seq{a}[m],b)}(\seq{t}[m]) \defsym \log(\seqx{a_\i \cdot \size{t_\i}}[m][+] + b)$ to also allow negative values for $b$ (under the condition that $\sum_i a_i + b \ge 1$) and our generalised \rulelettreecf\ rule can take advantage of these generalized resource functions (see \aref{fig:5} for a statement of the rule and the proof of its soundness as part of the proof of Theorem~\ref{t:2}).
%
%
%
%
%
%
%
%
%
%
%
%
%
%

%% file: conclusion.tex
In this paper, we present the first fully-automated \emph{expected amortised cost analysis} of self-adjusting data structures, that is, of \emph{randomised splay trees}, \emph{randomised splay heaps} and \emph{randomised meldable heaps}, which so far have only (semi-) manually been analysed in the literature.
In future work, we envision to extend our analysis to related probabilistic settings such as skip lists~\cite{Pugh90} and randomised binary search trees~\cite{MartinezR98}.
We note that adaptation of the framework developed in this paper to new benchmarks will likely require to identify new potential functions and the extension of the type-effect-system with typing rules for these potential functions.
Further, on more theoretical grounds we want to clarify the connection of the here proposed expected amortised cost analysis with Kaminski's \ert-calculus, cf.~\cite{KKMO:ACM:18}, and study whether the expected cost transformer is conceivable as a potential function.


%% file: appendix.tex
\section{Benchmark: Probabilistic Analysis of Binary Search Trees}
\label{sec:binary-search-trees}

We present a probabilistic analysis of a deterministic binary search tree, which offers the usual \bstcontains, \bstinsert, and \bstdelete\ operations, where \bstdelete\ uses \bstdelmax\ as a subroutine (the source code of all operations is given in Fig.~\ref{fig:appendix-coinsearchtree}).
We assume that the elements inserted, deleted and searched for are equally distributed;
hence, we conduct a probabilistic analysis by replacing every comparison with a coin toss of probability one half.
We will refer to the resulting data structure as Coin Search Tree in our benchmarks.
Our tool \atlas\ infers an logarithmic expected amortised cost for all operations, e.g., for \bstinsert\ and \bstdelmax\ we obtain
\begin{gather*}
  \sfrac{3}{2} \rk(t) + \sfrac{1}{2} \log(\size{t})  \geqslant c_{\text{\bstinsert}}(t) + \sfrac{3}{2} \rk(\text{\bstinsert}\ t) \\
  \sfrac{3}{2} \rk(t) + \sfrac{1}{2} \log(\size{t})  \geqslant c_{\text{\bstdelmax}}(t) + \sfrac{3}{2} \rk(\text{\bstdelmax}\ t)  \tkom
\end{gather*}
from which we obtain an expected amortised cost of $\sfrac{1}{2} \log(\size{t})$ for both functions.

\section{Omitted Definitions}
\subsection{Type System: Non-Probabilistic Part}

\begin{figure}[t]
\begin{mathpar}
\input{rules/typ/leaf.tex}
\and
\input{rules/typ/node.tex}
\and
\input{rules/typ/cmp.tex}
\and
\input{rules/typ/var.tex}
\and
\input{rules/typ/pair.tex}
\and
\input{rules/typ/ite.tex}
\and
\input{rules/typ/match.tex}
\and
\input{rules/typ/matchpair.tex}
\and
\input{rules/typ/letnotree.tex}
\and
\input{rules/typ/lettreecf.tex}
\and
\input{rules/typ/appticks.tex}
\end{mathpar}
To ease notation, we set $\veca \defsym \seq[1]{a}[m]$, $\vecb \defsym \seq[1]{b}[k]$ for vectors of indices $a_i, b_j \in \N$.
Further,
$i \in \{\upto{m}\}$, $j \in \{\upto{k}\}$, and $a,b,d \in \N$ and $c,e \in \Z$, where we recall that $c,e$ must be chosen such that
$\sum_i a_i + \sum_j b_j + c \geqslant 1$ resp. $a + b + c \geqslant 1$ as well as $\sum_j b_j + d + e \geqslant 1$ are satisfied.
Sequence elements of annotations, which are not constrained are set to zero.
Note that the conditions in \rulepair\ on coefficients are vacuously true,
if $\alpha_1 \not= \TreeShort$ and $\alpha_2 \not= \TreeShort$.
\caption{Syntax-Directed Type Rules: Non-Probabilisitc Part.}
\label{fig:5}
\end{figure}

\begin{figure}[t]
\begin{mathpar}
\input{rules/typ/wvar.tex}
\and
\input{rules/typ/share.tex}
\and
\input{rules/typ/w.tex}
\and
\input{rules/typ/shift.tex}
\end{mathpar}
%
\caption{Structural Type Rules: Non-Probabilisitc Part.}
\label{fig:5b}
\end{figure}

The non-probabilistic and structural typing rules are given in Figure~\ref{fig:5} and~\ref{fig:5b} respectively.
%
Let~$\Gamma$ be a variable context, $Q, Q'$ annotations and let~$e$ be an expression.
The typing rule for rule $\rulelettreecf$ makes use of the cost-free typing judgment
$\tjudgecfndt{\Gamma}{Q}{e}{\alpha}{Q'}$ that differs from the standard cost-free typing relation
$\tjudge{\Gamma}{Q}{e}{\alpha}{Q'}$ insofar that all probabilistic choices in~$e$ are replaced
by non-deterministic choices. We call the expression $e'$ obtained from $e$ through this adaption
the non-deterministic version of~$e$.

\subsection{Soundness Theorems}

The proof of the soundness theorems makes use of the following lemma, whose proof can be found in~\cite{hofmann2021typebased}.
\begin{lemma}
\label{lem:log-inequality}
Assume $\sum_i q_i \log a_i \geqslant q \log b$ for some rational numbers $a_i,b > 0$ and $q_i \geqslant q$.
Then, $\sum_i q_i \log (a_i + c) \geqslant q \log (b + c)$ for all $c \geqslant 1$.
\end{lemma}

\restatethmticknowsound*
\begin{proof}
We first deal with the case that $\Pi$ ends in a structural rule, cf.~Figure~\ref{fig:5b}:

\medskip
\emph{Case.} Suppose the last rule in $\Pi$ be of the following form:
\begin{equation*}
  \infer{
  \tjudge{\Gamma}{Q+K}{e}{\alpha}{Q'+K}
}{%
  \tjudge{\Gamma}{Q}{e}{\alpha}{Q'}
} \tkom
\end{equation*}
where $K \geqslant 0$. By SIH, we have that $\spotential{\Gamma}{Q} \geqslant c + \sum_{v \in \supp{\rst{\mu}}} \mu(v) \cdot \potential{v}{Q'}$, from
which we obtain
\begin{gather*}
  \spotential{\Gamma}{Q+K} = \spotential{\Gamma}{Q}+K \geqslant {}
  \\
  {} \geqslant c + \sum_{v \in \supp{\rst{\mu}}} \mu(v) \cdot \potential{v}{Q'} + K \geqslant
  c + \sum_{v \in \supp{\rst{\mu}}} \mu(v) \cdot \potential{v}{Q'+K}
  \tkom
\end{gather*}
as $\sum_{v \in \supp{\rst{\mu}}} \mu(v) \leqslant 1$ and $\potential{v}{Q'+K}=\potential{v}{Q'}+K$.

\medskip
\emph{Case.} Let $\Pi$ end in the following weakening rule applicaton
\begin{equation*}
  \input{rules/typ/w.tex}
  \tpkt
\end{equation*}
By SIH, we have $\spotential{\Gamma}{P} \geqslant c + \sum_{v \in \supp{\rst{\mu}}} \mu(p) \cdot \potential{v}{P'}$. Due to the
assumption of the \rulew\ rule, we have
\begin{align*}
  \spotential{\Gamma}{Q} & \geqslant \spotential{\Gamma}{P}
  \\
                          & \geqslant c + \sum_{v \in \supp{\rst{\mu}}} \mu(p) \cdot \potential{v}{P'}
  \\
                          & \geqslant c + \sum_{v \in \supp{\rst{\mu}}} \mu(p) \cdot \potential{v}{Q'}
                            \tpkt
\end{align*}

\medskip
\emph{Case.} $\ruleshare$ and $\rulewvar$ can be dealt with in the same way, we refer the reader to~\cite{hofmann2021typebased} for the details.

We now assume that $\Pi$ ends in a syntax-directed rule, cf.~Figure~\ref{fig:5}, and proceed by a case distinction on $e\sigma$, respectively the first step of $e\sigma \tomulti{c}_n \mu$:

\medskip
\emph{Case.} First we assume $e\sigma$ is a value.
By definition of $\tomulti{\cdot}$ we have $\mu = \{ {v} \}$ and $c = 0$.
There are several subcases to consider, eg. $e\sigma = \tree{t}{b}{u}$, $e\sigma = \leaf$, $e\sigma = \pair{e_1}{e_2}$, etc.
For these cases we can essentially proceed as in the non-probabilistic setting, cf.~\cite{hofmann2021typebased}.
%

Exemplarily, we consider the subcase where $e\sigma = \pair{e_1}{e_2}$.
\begin{equation*}
  \input{rules/typ/pair.tex}
  \tpkt
\end{equation*}
By definition and the constraints incorporated in \rulepair, we obtain
\begin{align*}
  \spotential{\typed{x_1}{\alpha_1}, \typed{x_2}{\alpha_2}}{Q} & \geqslant
                                                                 \potential{\pair{e_1}{e_2}}{Q'}
  \\
                                                               & = \sum_{v \in \{ {\pair{e_1}{e_2}}^1 \}} 1 \cdot \potential{v}{Q'}
                                                                 \tkom
\end{align*}
from which the claim follows.


\medskip
\emph{Case.} Consider
\begin{equation*}
  e = \flstk{match}\ x\ \flstk{with}
     \textup{\lstinline{|}}\ \leaf\ \arrow e_1\
     \textup{\lstinline{|}}\ \textup{\tree{x_1}{x_2}{x_3}}\ \arrow e_2
    \tkom
\end{equation*}
and suppose further $x\sigma = \tree{lu}{b}{v}$, that is,
$\{e\sigma\} \tomulti{0} \{e_2\} \tomulti{c}_{n-1} \mu$.
Because $\Pi$ ends with a syntax-directed rule, $\Pi$ must in fact end with an application of the \rulematch\ rule, ie.
\begin{equation*}
\input{rules/typ/match.tex} \tpkt
\end{equation*}
Note that $q_{m+1}$ denotes the coefficient of $\rk(x\sigma)$ in the definition of $\spotential{\Gamma,\typed{x}{\TreeShort}}{Q}$.
By definition and the constraints given in the rule, we obtain:
\begin{equation*}
  \spotential{\Gamma,\typed{x}{\TreeShort}}{Q} =
  \spotential{\Gamma, \typed{x_1}{\TreeShort}, \typed{x_2}{\BaseShort}, \typed{x_3}{\TreeShort}}{R}
\tpkt
\end{equation*}
By MIH we have
$\spotential{\Gamma, \typed{x_1}{\TreeShort}, \typed{x_2}{\BaseShort}, \typed{x_3}{\TreeShort}}{R}
\geqslant c + \sum_{v \in \supp{\rst{\mu}}} \mu(v) \potential{v}{Q'}$,
from which the case follows directly.

\medskip
\emph{Case.} Consider
\begin{equation*}
e = \flstk{let }x\flst{ = }e_1\flstk{ in }e_2.
\end{equation*}
In order to prove the claim for ${e\sigma} \tomulti{c}_n \mu$ we need to split the $n$-step derivation into $n_1$-step and $n_2$-step derivations for $e_1$ and $e_2$ with $n_1+n_2 + 1= n$, where the one step accounts for substituting the value to which $e_1$ has evaluated into $e_2$.

However, we cannot only consider one such split because evaluating $e_1$ to a normal form will in general need a different number of steps according to the probabilistic choices encountered in the derivation.
Hence, we will consider all possible splits.

For this, we consider $e_1\sigma \tomulti{c_i}_i \nu_i$ for all $0 \leqslant i \leqslant n$.
We recall that the $\nu_i$ are pointwise ordered on values, ie.\
we have ${\rst{\nu_i}} \leqslant {\rst{\nu_j}}$ for $i\leqslant j$.
Hence, we can define $\xi_i = {\rst{\nu_i}}-{\rst{\nu_{i-1}}}$ for all $0 < i \leqslant n$.
Note that for $w_i^{p_i} \in \xi_i$ we have that the probability that $e_1\sigma$ evaluates to the value $w_i$ in $i$ steps is exactly $p_i$.

Let $w$ be some value to which $e_1\sigma$ has evaluated to in $i$ steps.
We then note that $\flstk{let }x\flst{ = }w\flstk{ in }e_2 \tomulti{0}_1 \{ {e_2[x \mapsto w]} \}$.
Thus, we can apply the MIH to $e_2\sigma[x \mapsto w]$ and obtain that
\begin{gather*}
  e_2\sigma[x \mapsto w] \tomulti{c_{w,i}}_{n-i-1} \mu_{w,i}\\
  {} \Rightarrow \potential{{\sigma[x \mapsto w]};{\Delta, \typed{x}{\alpha}}}{R} \geqslant c_{w,i} + \sum_{v \in \supp{{\mu_{w,i}\restriction_{V}}}} \mu_{w,i}(v) \cdot \potential{v}{Q'} \tag{$\dag$}
  \tkom
\end{gather*}
for suitably defined distributions $\mu_{w,i}$ and costs $c_{w,i}$.
We now consider the SIH applied to $e_1\sigma \tomulti{c_1}_n \nu$, ie.\ we have that
\begin{equation*}
  e_1\sigma \tomulti{c_1}_n \nu \Rightarrow
  \potential{{\sigma};{\Gamma}}{P} \geqslant c_1 + \sum_{w \in \supp{{\nu}\restriction_{V}}} \nu(w) \cdot \potential{w}{P'} \tag{$\ddag$}
  \tpkt
\end{equation*}
By the definition of the $\xi_i$ we have that $\rst{\nu} = \sum_{i=1}^{n} \xi_i$.
We then consider ${e\sigma} \tomulti{c}_n \mu$.
We observe that $\rst{\mu} = \sum_{i=1}^{n} \sum_{w^{p_i} \in \xi_i} p_i\cdot \rst{\mu_{w,i}}$ and $c = c_1 + \sum_{i=1}^{n} \sum_{w^{p_i} \in \xi_i} p_i\cdot c_{w,i}$ for distributions $\mu_{w,i}$ and costs $c_{w,i}$ defined as above.
Further, we will establish below that
\begin{multline*}
      \spotential{\Gamma,\Delta}{Q} + \sum_{w \in \supp{{\nu}\restriction_{V}}} \nu(w) \cdot \potential{w}{P'} \ge  \\
   \potential{{\sigma};{\Gamma}}{P} +
    \sum_{w \in \supp{{\nu}\restriction_{V}}} \nu(w) \cdot
                                  \potential{{\sigma[x \mapsto w]};{\Delta, \typed{x}{\alpha}}}{R} \tag{$\star$} \tpkt
\end{multline*}

We finally calculate using $(\dag)$, $(\ddag)$ and $(\star)$ that
\begin{align*}
  \spotential{\Gamma,\Delta}{Q} & \geqslant c_1 + \sum_{w \in \supp{{\nu}\restriction_{V}}} \nu(w) \cdot
                                  \potential{{\sigma[x \mapsto w]};{\Delta, \typed{x}{\alpha}}}{R}
  \\
  & = c_1 + \sum_{i=1}^{n} \sum_{w \in \supp{\xi_i}} \xi_i(w) \cdot  \potential{{\sigma[x \mapsto w]};{\Delta, \typed{x}{\alpha}}}{R}
  \\
  & \geqslant c_1 + \sum_{i=1}^{n} \sum_{w^{p_i} \in \xi_i} p_i \cdot (c_{w,i} + \sum_{v \in \supp{{\mu_{w,i}\restriction_{V}}}} \mu_{w,i}(v) \potential{v}{Q'})
  \\
  & = c_1 + \sum_{i=1}^{n} \sum_{v_i^{p_i} \in \xi_i} p_i \cdot c_{w,i}\\
  & \quad \quad  + \sum_{i=1}^n \sum_{w^{p_i} \in \xi_i} p_i \cdot \sum_{v \in \supp{{\mu_{w,i}}\restriction_{V}}} \mu_{w,i}(v) \potential{v}{Q'}\\
  & = c + \sum_{v \in \supp{{\mu}\restriction_{V}}} \mu(v) \cdot \potential{v}{Q'})
    \tpkt
\end{align*}

In the first line, we employ property~$(\star)$ together with the observation that
$(\dag)$ implies $\potential{{\sigma};{\Gamma}}{P} - \sum_{w \in \supp{{\nu}\restriction_{V}}} \nu(w) \cdot \potential{w}{P'} \geqslant c_1$.

It remains to establish $(\star)$.
For this we proceed by a case distinction on whether $e_1$ is of tree type, \ie{}, whether the rule $\Pi$ ends
in an application of the $\rulelettreecf$- or of the $\ruleletnotree$-rule.
We treat the simpler case first and consider that $e_1$ is not of tree type.
Then, $\Pi$ ends in an application of the $\ruleletnotree$-rule, ie.\
\begin{equation*}
\small
\input{rules/typ/letnotree.tex} \tpkt
\end{equation*}
We note that $(\star)$ follows directly from the constraints in the $\ruleletnotree$ together with the fact that $\nu$ is a (sub-)distribution.

Finally, we now suppose that $e_1$ is of type tree.
Then, the type derivation $\Pi$ ends
in an application of the \rulelettreecf-rule.
\begin{equation*}
\small
\infer[\rulelettreecf]{
  \tjudge{\Gamma, \Delta}{Q}{\flstk{let }x\flst{ = }e_1\flstk{ in }e_2}{\beta}{Q'}
}{%
  \begin{minipage}[b]{70ex}
    \centering
    $\tjudge{\Gamma}{P}{e_1}{\TreeShort}{P'}$
    \hfill
    $\forall {\vecb \ne \vec{0},d \ne 0 } \ \left( \tjudgecfndt{\Gamma}{P^{(\vecb,d,e)}}{e_1}{\TreeShort}{{P'}^{(\vecb,d,e)}} \right)$
    \\[1ex]
    $\tjudge{\Delta, \typed{x}{\TreeShort}}{R}{e_2}{\beta}{Q'}$
  \end{minipage}
}
\tkom
\end{equation*}
where we have elided all arithmetic constraints for readability.
By definition and due to the constraints expressed in the typing rule, we have that
\begin{align*}
  \spotential{\Gamma,\Delta}{Q} & =
                                  \sum_i q_i \rk(t_i) + \sum_j q_j \rk(u_j) + {}
                                  \\
  & \qquad \sum_{\veca, \vecb, c} q_{(\veca,\vecb,c)} \log(\veca\size{\vect} + \vecb\size{\vecu} + c)
  \\
  \spotential{\Gamma}{P} & = \sum_i q_i \rk(t_i) + \sum_{\veca, c} q_{(\veca, \vec{0}, c)} \log (\veca \size{\vect} + c)
  \\
  \potential{w}{P'} &= r_{k+1} \rk(w) + \sum_{d,e} r_{(\vec{0},d,e)} \log (d\size{w} +e)
  \\
  \spotential{\Gamma}{P^{(\vecb,d,e)}} &= \sum_{\veca, c} p^{(\vecb,d,e)}_{(\veca,c)} \log (\veca \size{\vect} + c)
  \\
  \potential{w}{{P'}^{(\vecb,d,e)}} &= {p'}^{(\vecb,d,e)}_{(d,\max{(e,0)})} \log (d\size{w} + \max{\{e,0\}})
  \\
      \potential{{\sigma[x \mapsto w]};{\Delta,\typed{x}{\TreeShort}}}{R} & = \sum_{j} q_j \rk(u_j) + r_{k+1} \rk(w) + {}
      \\
      & \qquad \sum_{\vecb, d, e} r_{(\vecb,d,e)} \log(\vecb\size{\vec{u}}+d\size{w}+e)
\tkom
\end{align*}
where we set $\vect \defsym \seq[1]{t}[m]$ and $\vec{u} \defsym \seq[1]{u}[k]$, denoting the substitution instances of the variables in $\Gamma$, $\Delta$, respectively.
(We recall that the well-definedness of $\log$ is implicitly assumed:
$\sum_i a_i + \sum_j b_j + c \geqslant 1$ resp.\ $a + b + c \geqslant 1$ as well as $\sum_j b_j + d + e \geqslant 1$ are satisfied.)

%
We now consider some $w \in \supp{\rst{\nu}}$.
By definition of non-deterministic (cost-free) version of $e_1$ there is a small-step derivation $e_1\sigma \tomulti{0}_n \{w^1\}$.
Due to the cost-free typing constraints  $\tjudgecfndt{\Gamma}{P^{(\vecb,d,e)}}{e_1}{\TreeShort}{{P'}^{(\vecb,d,e)}}$
and the SIH applied to $e_1\sigma \tomulti{0}_n \{w^1\}$ we have that
\begin{equation*}
\spotential{\Gamma}{P^{(\vecb,d,e)}} \geqslant \potential{w}{{P'}^{(\vecb,d,e)}}
\end{equation*}
for all $\vecb \not= \vec{0},d \neq 0$, ie. that
\begin{equation*}\label{eq:soundness-alt}
    \sum_{\veca,c} p^{(\vecb,d,e)}_{(\veca,c)} \log (\veca \size{\vect} + c) \geqslant
    {p'}^{(\vecb,d,e)}_{(d,\max{\{e,0\}})} \log (d\size{w} + \max{\{e,0\}})
  \tpkt
\end{equation*}

Due to the conditions $\sum_{(a,c)} p^{(\vecb,d,e)}_{(\veca,c)} \geqslant {p'}^{(\vecb,d,e)}_{(d,\max{\{e,0\}}}$, ${p'}^{(\vecb,d,e)}_{(d',e')} = 0$ for all $(d',e') \neq (d,\max{\{e,0\}})$, and
$p^{(\vecb,d,e)}_{(\veca,c)} \neq 0$ implies that
${p'}^{(\vecb,d,e)}_{(d,\max{\{e,0\}}} \leqslant p^{(\vecb,d,e)}_{(\veca,c)}$ for all $\veca \neq 0$,
we can apply Lemma~\ref{lem:log-inequality} to Equation~\eqref{eq:soundness-alt} and obtain
\begin{equation*}
  \sum_{\vec{a}, c} p^{(\vecb,d,e)}_{(\veca,c)} \log (\veca \size{\vect} + \vecb \size{\vecu} + c - \max{\{-e,0\}}) \geqslant
  {p'}^{(\vecb,d,e)}_{(d,\max{\{e,0\}})} \log (\vecb \size{\vecu} + d\size{w} + e)
  \tpkt
\end{equation*}
Note that if $e \geqslant 0$ then  $-\max{\{-e,0\}}=0$ and if $e < 0$ then $-\max{\{-e,0\}}=e$.
Thus in the former we add the sum $\vecb \size{\vecu}$ to both sides of the  inequality~\eqref{eq:soundness-alt}, while in the second case we add $\vecb \size{\vecu} + e$.

Note the conditions
\begin{inparaenum}[(i)]
  \item $q_{(\veca,\vecb,c)} = \sum_{(d,e)} p^{(\vecb,d,e)}_{(\veca,c+\max\{-e,0\})}$ and
  \item $r_{(\vec{b},d,e)} = {p'}^{(\vecb,d,e)}_{(d,\max{\{e,0\}})}$
\end{inparaenum}
for all $\vec{b} \ne \vec{0}, \vec{a} \ne \vec{0}, d \neq 0$.

Thus, we can sum up those equations for all $\vec{b} \ne \vec{0}, \vec{a} \ne \vec{0}, d \neq 0$ and obtain that
\begin{equation*}
  \sum_{\vec{b} \ne \vec{0}, \vec{a} \ne \vec{0}} q_{(\veca,\vecb,c)} \log (\veca \size{\vect} + \vecb \size{\vecu} + c) \geqslant
    \sum_{\vec{b} \ne \vec{0},d \ne 0} r_{(\vec{b},d,e)} \log (\vecb \size{\vecu} + d\size{w} + e)
  \tpkt
\end{equation*}
Because the above equation holds for any $w \in \supp{\rst{\nu}}$ we can deduce that
\begin{gather*}
  \sum_{\veca \ne \vec{0}, \vecb \ne \vec{0}, c} q_{(\veca,\vecb,c)} \log (\veca \size{\vect} + \vecb \size{\vecu} + c) \geqslant {} \notag
  \\ {} \geqslant \sum_{w \in \supp{\rst{\nu}}} \nu(w) \cdot \left(  \sum_{\veca \ne \vec{0}, \vec{b} \ne \vec{0}, c} q_{(\veca,\vecb,c)} \log (\veca \size{\vect} + \vecb \size{\vecu} + c) \right)
  \\ {} \geqslant \sum_{w \in \supp{\rst{\nu}}} \nu(w) \cdot \bigl(
    \sum_{\vecb \ne \vec{0},d \ne 0,e} r_{(\vec{b},d,e)} \log (\vecb \size{\vecu} + d\size{w} + e)
    \bigr)
  \tkom
\end{gather*}
using that $\nu$ is a (sub-)distribution, \ie{}, that $\sum_{w \in \supp{\rst{\nu}}} \nu(w) \le 1$.
We now note that $(\star)$ follows directly from  the above inequality and the constraints in the $\rulelettreecf$ together with the fact that $\nu$ is a (sub-)distribution.

\medskip
\emph{Case.} Let $e$ be a conditional and assume the last rule in $\Pi$ is of the following form:
\begin{equation*}
  \input{rules/typ/ite.tex}
  \tpkt
\end{equation*}
(The case where the condition is performed non-deterministically, is treated analogousloy.) By assumption, we either have
\begin{enumerate}[(i)]
  \item $\{e\sigma^1 \} = \{ \flstk{if }\true \flstk{ then }e_1\flstk{ else }e_2 \} \toop{0} \{ e_1 \} \tomulti{c}_{n-1}  \mu$ or
  \item $\{e\sigma^1 \} = \{ \flstk{if }\false \flstk{ then }e_1\flstk{ else }e_2 \} \toop{0} \{ e_2 \} \tomulti{c}_{n-1}  \mu$.
  \end{enumerate}
  In both case MIH yields that $\spotential{\Gamma}{Q} \geqslant c + \sum_{v \in \supp{\rst{\mu}}} \mu(p) \cdot \potential{v}{Q'}$, from which the theorem follows as $\spotential{\Gamma, \typed{x}{\BoolShort}}{Q} = \spotential{\Gamma}{Q}$ by definition.

\emph{Case.} Let $e$ be ticking statement and let the last rule in $\Pi$ be of the following form:
\begin{equation*}
  \input{rules/typ/ticking.tex}
  \tpkt
\end{equation*}
By definition, we have $\{ \tick{e\sigma} \} \toop{\sfrac{a}{b}} \{ e\sigma \} \tomulti{c-\sfrac{a}{b}}_{n-1}  \mu$ and by MIH, we have
\begin{equation*}
  \spotential{\Gamma}{Q} \geqslant (c-\sfrac{a}{b}) + \sum_{v \in \supp{\rst{\mu}}} \mu(p) \cdot \potential{v}{Q'}
  \tpkt
\end{equation*}
Hence $\spotential{\Gamma}{Q+\sfrac{a}{b}} \geqslant c + \sum_{v \in \supp{\rst{\mu}}} \mu(p) \cdot \potential{v}{Q'}$.

\medskip
\emph{Case}. Let $e$ be a probabilistic branching statement, that is
\begin{equation*}
  e\sigma = \flstk{if } \coin[$a$][$b$] \flstk{ then }e_1\flstk{ else }e_2
  \tkom
\end{equation*}
and let the last rule in $\Pi$ be of the following form
\begin{equation*}
  \small
  \input{rules/typ/itecoin.tex}
  \tpkt
\end{equation*}
By definition, we have $\{ {e\sigma} \} \toop{0} \{ e_1^p, e_2^{1-p} \} \tomulti{c}_{n-1}  \mu$.
By definition
of $\tomulti{\cdot}$ there exists (sub)distribution $\mu_1$, $\mu_2$ s.t. $\{e_1\} \tomulti{c_1}_{m_1}  \mu_1$
and $\{e_2\} \tomulti{c_2}_{m_2}  \mu_2$, where $m_1,m_2 < n$, $\mu = p \cdot \mu_1 + (1-p) \cdot \mu_2$
and $c = p \cdot c_1 + (1-p) \cdot c_2$. Further by SIH, we conclude
\begin{inparaenum}[(i)]
  \item $\spotential{\Gamma}{Q_1} \geqslant c_1 + \sum_{v \in \supp{\rst{\mu_1}}} \mu_1(p) \cdot \potential{v}{Q'}$ and
  \item $\spotential{\Gamma}{Q_2} \geqslant c_2 + \sum_{v \in \supp{\rst{\mu_2}}} \mu_2(p) \cdot \potential{v}{Q'}$.
\end{inparaenum}
Hence, we obtain
\begin{align*}
  \spotential{\Gamma}{Q} & = \spotential{\Gamma}{p \cdot Q_1 + (1-p) \cdot Q_2}
  \\
                         & = p \cdot \spotential{\Gamma}{Q_1} + (1-p) \cdot \spotential{\Gamma}{Q_2}
  \\
                         & \geqslant c_1 + c_2 + (p \cdot \sum_{v^{q} \in \rst{\mu_1}} q \cdot \potential{v}{Q'}) +
                           ((1-p) \cdot \sum_{v^{q} \in \rst{\mu_2}} q \cdot \potential{v}{Q'})
  \\
                         & = c + \sum_{v^{q} \in \rst{\mu_1}} p \cdot q \cdot \potential{v}{Q'} +
                           \sum_{v^{q'} \in \rst{\mu_2}} (1-p) \cdot q \cdot \potential{v}{Q'}
  \\
                         & = c + \sum_{v^{q} \in p \cdot \rst{\mu_1} + (1-p) \cdot \rst{\mu_2}} q \cdot \potential{v}{Q'} = c + \sum_{v^q \in \rst{\mu}}  q \cdot \potential{v}{Q'}
                           \tkom
\end{align*}
from which we conclude the case.

\medskip
\emph{Case}. We consider the application rules~\ruleapp\ and~\ruleappcf\ and restrict our argument to the former, as the
proof for the cost-free variant is similar, but simpler. We consider the costed typing
\begin{equation*}
  \small
  \input{rules/typ/appticks.tex}
  \tpkt
\end{equation*}
Let $f(\seq{x}[k]) = e \in \Program$, as $\Program$ is well-typed, we have $\tjudge{\Gamma}{P}{e}{\beta}{P'}$ and
$\tjudgecf{\Gamma}{Q}{e}{\beta}{Q'}$ by assumption. Further, by definition $\{ e\sigma \} \toop{0} e \tomulti{c}_{n-1} \mu$.
We conclude by MIH that $\spotential{\Gamma}{P} \geqslant c + \sum_{v \in \supp{\rst{\mu}}} \mu(p) \cdot \potential{v}{P'}$ and
$\spotential{\Gamma}{Q} \geqslant \sum_{v \in \supp{\rst{\mu}}} \mu(p) \cdot \potential{v}{Q'}$.
Hence
\begin{align*}
  \spotential{\Gamma}{P + K \cdot Q} & = \spotential{\Gamma}{P} + K \cdot  \spotential{\Gamma}{Q}
  \\
                                     & \geqslant c + \sum_{v \in \supp{\rst{\mu}}} \mu(p) \cdot \potential{v}{P'} + K \cdot \sum_{v \in \supp{\rst{\mu}}} \mu(p) \cdot \potential{v}{Q'}
  \\
                                     & = c +  \sum_{v \in \supp{\rst{\mu}}} \mu(p) \cdot ( \potential{v}{P'} + K \cdot \potential{v}{Q'} )
  \\
                                     & = c + \sum_{v \in \supp{\rst{\mu}}} \mu(p) \cdot \potential{v}{P' + K \cdot Q'}
                                       \tpkt
\end{align*}
\qed
\end{proof}

\restatethmtickdefersound*
\begin{proof}
The setup (and most of the cases) of this proof follow the proof of Theorem~\ref{t:1}:
It suffices to prove for every $n \geqslant 0$ that
\begin{equation*}
  \evaln{\sigma}{n}{c}{e}{\mu} \mu \Rightarrow \spotential{\Gamma}{Q} \geqslant c + \Expect{\mu}{\lambda v. \potential{v}{Q'}}
\end{equation*}
We proceed by main induction on $n$ ---which we will call main induction hypothesis (MIH)--- and side-induction on the length of the type derivation $\Pi$ of $\tjudge{\Gamma}{Q}{e}{\alpha}{Q'}$ ---which we will call side induction hypothesis (SIH).

For the majority of the cases, the arguments can be easily suited from those employed in proof
of Theorem~\ref{t:1}.
Thus, we only consider a restricted set of cases that may be of independent interest.

We now consider $\evaln{\sigma}{n+1}{c}{e}{\mu}$ for some $n \geqslant 0$ and the type derivation $\Pi$ of $\tjudge{\Gamma}{Q}{e}{\alpha}{Q'}$.
The cases where $\Pi$ ends in a structural rule, cf.~Figure~\ref{fig:5b}, can be dealt with in the same way as in the proof of Theorem~\ref{t:1}.

We now assume that $\Pi$ ends in a syntax-directed rule, cf.~Figure~\ref{fig:5}, and proceed by a case distinction on $e\sigma$, respectively the first step of $\evaln{\sigma}{n+1}{c}{e}{\mu}$:

\medskip
\emph{Case.} Consider
\begin{equation*}
e = \flstk{let }x\flst{ = }e_1\flstk{ in }e_2.
\end{equation*}
Let $w$ be some value.
We apply the MIH to $\evaln{\sigma[x \mapsto w]}{n}{c_w}{e_2}{\mu_w}$ and obtain that
\begin{gather*}
  \evaln{\sigma[x \mapsto w]}{n}{c_w}{e_2}{\mu_w}\\
  {} \Rightarrow \potential{{\sigma[x \mapsto w]};{\Delta, \typed{x}{\alpha}}}{R} \geqslant c_w + \sum_{v \in \supp{{\mu_w\restriction_{V}}}} \mu_w(v) \cdot \potential{v}{Q'} \tag{$\dag$}
  \tkom
\end{gather*}
for suitably defined distributions $\mu_w$ and costs $c_w$.
We further apply the MIH to $\evaln{\sigma}{n}{c_1}{e_1}{\nu}$, ie.\ we obtain that
\begin{equation*}
  \evaln{\sigma}{n}{c_1}{e_1}{\nu} \Rightarrow
  \potential{{\sigma};{\Gamma}}{P} \geqslant c_1 + \sum_{w \in \supp{{\nu}\restriction_{V}}} \nu(w) \cdot \potential{w}{P'} \tag{$\ddag$}
  \tpkt
\end{equation*}
Further, we will establish below that
\begin{multline*}
      \spotential{\Gamma,\Delta}{Q} + \sum_{w \in \supp{{\nu}\restriction_{V}}} \nu(w) \cdot \potential{w}{P'} \ge  \\
   \potential{{\sigma};{\Gamma}}{P} +
    \sum_{w \in \supp{{\nu}\restriction_{V}}} \nu(w) \cdot
                                  \potential{{\sigma[x \mapsto w]};{\Delta, \typed{x}{\alpha}}}{R} \tag{$\star$} \tpkt
\end{multline*}

We finally calculate using $(\dag)$, $(\ddag)$ and $(\star)$ that
\begin{align*}
  \spotential{\Gamma,\Delta}{Q} & \geqslant c_1 + \sum_{w \in \supp{{\nu}\restriction_{V}}} \nu(w) \cdot
                                  \potential{{\sigma[x \mapsto w]};{\Delta, \typed{x}{\alpha}}}{R}
  \\
  & \geqslant c_1 + \sum_{w \in \supp{{\nu}\restriction_{V}}} \nu(w) \cdot ( c_w + \sum_{v \in \supp{{\mu_w\restriction_{V}}}} \mu_w(v) \cdot \potential{v}{Q'})
  \\
  & = c_1 + \sum_{w \in \supp{{\nu}\restriction_{V}}} \nu(w) \cdot c_w +
  \\
  & \quad \quad  + \sum_{w \in \supp{{\nu}\restriction_{V}}} \sum_{v \in \supp{{\mu_w\restriction_{V}}}} \nu(w) \cdot \mu_w(v) \cdot \potential{v}{Q'}
  \\
  & = c + \sum_{v \in \supp{{\mu}\restriction_{V}}} \mu(v) \cdot \potential{v}{Q'})
    \tpkt
\end{align*}
where we have used for the last equality that $\mu = \sum_{w \in \supp{\nu}} \nu(w) \cdot \mu_w$ and
$c= c_2 + \sum_{w \in \supp{\nu}} \nu(w) \cdot c_w$ according to the definition of the big-step semantics.

We finally note that $(\star)$ can be established in the same way as in the proof of Theorem~\ref{t:1} (for both the $\ruleletnotree$- and the $\rulelettreecf$-rule case).

\medskip
\emph{Case.} Let $e$ be a probabilistic branching statement, that is
\begin{equation*}
  e\sigma = \flstk{if } \coin[$a$][$b$] \flstk{ then }e_1\flstk{ else }e_2
  \tkom
\end{equation*}
and let the last rule in $\Pi$ be of the following form
\begin{equation*}
  \small
  \input{rules/typ/itecoin.tex}
  \tpkt
\end{equation*}
By definition, there exists distributions $\mu_1$ and $\mu_2$ such that $\evaln{\sigma}{n}{c_1}{e_1}{\mu_1}$,
$\evaln{\sigma}{n}{c_2}{e_2}{\mu_2}$, $\mu = p \cdot \mu_1 + (1-p) \cdot \mu_2$ and $c = p \cdot c_1 + (1-p) \cdot c_2$.
By MIH, we conclude
\begin{equation*}
  \spotential{\Gamma}{Q} \geqslant c_1 + \Expect{\mu_1}{\lambda v. \potential{v}{Q'}}
  \qquad
  \spotential{\Gamma}{Q_2} \geqslant c_2 + \Expect{\mu_2}{\lambda v. \potential{v}{Q'}}
\end{equation*}
Hence, we obtain
\begin{align*}
  \spotential{\Gamma}{Q} & = \spotential{\Gamma}{p \cdot Q_1 + (1-p) \cdot Q_2}
  \\
                         & = p \cdot \spotential{\Gamma}{Q_1} + (1-p) \cdot \spotential{\Gamma}{Q_2}
  \\
                         & \geqslant p \cdot c_1 + (1-p) \cdot c_2 + (p \cdot \sum_{v^{q} \in \mu_1} q \cdot \potential{v}{Q'}) +
                           ((1-p) \cdot \sum_{v^{q} \in \mu_2} q \cdot \potential{v}{Q'})
  \\
                         & = c + \sum_{v^{q} \in \mu_1} p \cdot q \cdot \potential{v}{Q'} +
                           \sum_{v^{q'} \in \mu_2} (1-p) \cdot q \cdot \potential{v}{Q'}
  \\
                         & = c + \sum_{v^{q} \in p \cdot \mu_1 + (1-p) \cdot \mu_2} q \cdot \potential{v}{Q'} = c + \sum_{v^q \in \mu}  q \cdot \potential{v}{Q'}
                           \tkom
\end{align*}
from which we conclude the case.

\medskip
\emph{Case.}
Suppose the first step in the derivation of $\evaln{\sigma}{n+1}{c}{e}{\mu}$ is
\begin{equation*}
  \input{rules/sem/tickB.tex}
  \tpkt
\end{equation*}
and that $\Pi$ ends with the rule
\begin{equation*}
  \input{rules/typ/tickingast.tex}
  \tpkt
\end{equation*}
By MIH, we obtain $\spotential{\Gamma}{Q} \geqslant c + \Expect{\mu}{\lambda v. \potential{v}{Q'}}$,
from which we conclude that
\begin{align*}
  \spotential{\Gamma}{Q} & \geqslant c + \sum_{v \in \supp{\mu}} \mu(v) \cdot \potential{v}{Q'}
  \\
                         & = c + \sum_{v \in \supp{\mu}} \mu(v) \cdot
                           \bigl( \potential{v}{Q'} - \sfrac{a}{b} + \sfrac{a}{b} \bigr)
  \\
                         & = c + \sum_{v \in \supp{\mu}} \mu(v) \sfrac{a}{b} +
                           \sum_{v \in \supp{\mu}} \mu(v) \cdot \bigl( \potential{v}{Q'} - \sfrac{a}{b} \bigr)
  \\
                         & = c + \prob{\mu} \cdot  \sfrac{a}{b}  +
                           \sum_{v \in \supp{\mu}} \mu(v) \cdot \potential{v}{Q' - \sfrac{a}{b}}
  \\
                         &= c + \prob{\mu} \cdot  \sfrac{a}{b} + \Expect{\mu}{\lambda v. \potential{v}{Q'-\sfrac{a}{b}}}
                           \tpkt
\end{align*}
Here, we exploit the definition of $Q'-\sfrac{a}{b}$ and the definition of $\prob{\mu}$
in the second-to-last line and the definition of expectations in the last line.
\end{proof}

\section{Function Definitions}

Below, we use a notation for ticks that is easier to type with standard keyboard layouts,
\ie{} the tilde symbol followed by cost and the subexpression, \texttt{$\sim$ $a$/$b$ $e$},
instead of a tick mark and cost in the superscript, $\tick[$a$][$b$]{e}$.

\lstset{basicstyle=\scriptsize\ttfamily}
\newcommand{\inline}[1]{{\ttfamily#1}}

\paragraph{Randomised Splay Trees.} The benchmark comprises the functions
\inline{splay},\\
\inline{splay\_max},
\inline{insert}, and
\inline{delete}, see Figure~\ref{fig:appendix-randsplaytree}.

\paragraph{Randomized Splay Heaps.} The benchmark comprises the functions
\inline{insert}, and
\inline{delete\_min}, see Figure~\ref{fig:appendix-randsplayheap}.

\paragraph{Meldable Heaps.} The benchmark comprises the functions
\inline{meld},
\inline{insert}, and
\inline{delete\_min}, see Figure~\ref{fig:appendix-randmeldableheap}.

\paragraph{Coin Search Tree.} The benchmark comprises the functions
\inline{insert},
\inline{contains},
\inline{delete}, and
\inline{delete\_max}, see Figure~\ref{fig:appendix-coinsearchtree}.

\paragraph{Tree.} The benchmark comprises the function
\inline{descend}, see Figure~\ref{fig:appendix-tree}.

\begin{figure}[p]
\begin{lstlisting}[mathescape,breakatwhitespace,firstline=1,lastline=52]
splay a t = match t with
  | node cl c cr -> if a == c
    then (node cl c cr)
    else if a < c
      then match cl with
        | leaf         -> node leaf c cr
        | node bl b br -> if a == b
          then node (node bl a br) c cr  
          else if a < b
            then match bl with
              | leaf -> node leaf b (node br c cr)
              | bl   -> match ~ 1/2 splay a bl with
                | node al _ ar -> if coin
                  then ~ 1/2 node al a (node ar b (node br c cr))
                  else       node (node (node al a ar) b br) c cr
            else match br with
              | leaf -> (node bl b (node leaf c cr))
              | br   -> match ~ 1/2 splay a br with
                | node al _ ar -> if coin
                  then ~ 1/2 node (node bl b al) a (node ar c cr)
                  else       node (node bl b (node al a ar)) c cr
      else match cr with
        | leaf        -> (node cl c leaf)
        | node bl b br -> if a == b
          then (node cl c (node bl a br))  
          else if a < b
            then match bl with
              | leaf -> (node (node cl c leaf) b br)
              | bl   -> match ~ 1/2 splay a bl with
                | node al _ ar -> if coin
                  then ~ 1/2 node (node cl c al) a (node ar b br)
                  else       node cl c (node (node al a ar) b br)
            else match br with
              | leaf -> (node (node cl c bl) b leaf)
              | br   -> match ~ 1/2 splay a br with
                | node al _ ar -> if coin
                  then ~ 1/2 node (node (node cl c bl) b al) a ar
                  else       node cl c (node bl b (node al a ar))

splay_max z t = match t with
  | leaf      -> (leaf, z)
  | node l b r -> match r with
    | leaf        -> (node l b leaf, b)
    | node rl c rr -> match rr with
      | leaf -> (node (node l b rl) c leaf, c)
      | rr   -> match ~ 1/2 splay_max z rr with
        | (r1, max) -> match r1 with
          | leaf          -> (leaf, z)
          | node rrl1 x xa -> if coin
            then ~ 1/2 (node (node (node l b rl) c rrl1) x xa, max)
            else       (node l b (node rl c (node rrl1 x xa)), max) 
\end{lstlisting}
\end{figure}
\begin{figure}[p]
\ContinuedFloat
\begin{lstlisting}[mathescape,breakatwhitespace,firstline=53,lastline=90,firstnumber=53]
insert a t = match t with
  | node cl c cr -> if a == c
    then (node cl c cr)
    else if a < c
      then match cl with
        | leaf         -> node (node leaf a leaf) c cr
        | node bl b br -> if a == b
          then node (node bl a br) c cr 
          else if a < b
            then match bl with
              | leaf -> node (node leaf a leaf) b (node br c cr)
              | bl   -> match ~ 1/2 insert a bl with
                | node al _ ar -> if coin
                  then ~ 1/2 node al a (node ar b (node br c cr))
                  else       node (node (node al a ar) b br) c cr
            else match br with
              | leaf -> node bl b (node (node leaf a leaf) c cr)
              | br   -> match ~ 1/2 insert a br with
                | node al _ ar -> if coin
                  then ~ 1/2 node (node bl b al) a (node ar c cr)
                  else       node (node bl b (node al a ar)) c cr
      else match cr with
        | leaf         -> node cl c (node leaf a leaf)
        | node bl b br -> if a == b
          then node cl c (node bl a br) 
          else if a < b
            then match bl with
              | leaf -> node (node cl c (node leaf a leaf)) b br
              | bl   -> match ~ 1/2 insert a bl with
                | node al _ ar -> if coin
                  then ~ 1/2 node (node cl c al) a (node ar b br)
                  else       node cl c (node (node al a ar) b br)
            else match br with
              | leaf -> node (node cl c bl) b (node leaf a leaf)
              | br   -> match ~ 1/2 insert a br with
                | node al _ ar -> if coin
                  then ~ 1/2 node (node (node cl c bl) b al) a ar
                  else       node cl c (node bl b (node al a ar))
\end{lstlisting}
\end{figure}
\begin{figure}[p]
\ContinuedFloat
\begin{lstlisting}[mathescape,breakatwhitespace,firstline=92,firstnumber=92]
delete z a t = match t with
  | node cl c cr -> if a == c
    then match splay_max z cl with
      | (cl1, max) -> node cl1 max cr
    else if a < c
      then match cl with
        | leaf         -> node leaf c cr
        | node bl b br -> if a == b
          then match splay_max z bl with
            | (bl1, max) -> node (node bl1 max br) c cr
          else if a < b
            then match bl with
              | leaf -> node leaf b (node br c cr)
              | bl   -> match ~ 1/2 delete z a bl with
                | node al _ ar -> if coin
                  then ~ 1/2 node al a (node ar b (node br c cr))
                  else       node (node (node al a ar) b br) c cr
            else match br with
              | leaf -> node bl b (node leaf c cr)
              | br   -> match ~ 1/2 delete z a br with
                | node al _ ar -> if coin
                  then ~ 1/2 node (node bl b al) a (node ar c cr)
                  else       node (node bl b (node al a ar)) c cr
      else match cr with
        | leaf         -> node cl c leaf
        | node bl b br -> if a == b
          then match splay_max z bl with
            | (bl1, max) -> node cl c (node bl1 max br)
          else if a < b
            then match bl with
              | leaf -> node (node cl c leaf) b br
              | bl   -> match ~ 1/2 delete z a bl with
                | node al _ ar -> if coin
                  then ~ 1/2 node (node cl c al) a (node ar b br)
                  else       node cl c (node (node al a ar) b br)
            else match br with
              | leaf -> node (node cl c bl) b leaf
              | br   -> match ~ 1/2 delete z a br with
                | node al _ ar -> if coin
                  then ~ 1/2 node (node (node cl c bl) b al) a ar
                  else       node cl c (node bl b (node al a ar))
\end{lstlisting}
\setcounter{figure}{1}
\caption{Module \inline{RandSplayTree}.}
\label{fig:appendix-randsplaytree}
\end{figure}
\begin{figure}[t]
\begin{lstlisting}[mathescape]
insert d t = match t with
  | node tab ab tbc -> if ab <= d
    then match tbc with
      | leaf         -> node tab ab (node leaf d leaf)
      | node tb b tc -> if b <= d
        then match ~ 1/2 insert d tc with 
          | node tc1 c tc2 -> if coin
            then ~ 1/2 node (node (node tab ab tb) b tc1) c tc2
            else       node tab ab (node tb b (node tc1 c tc2))
        else match ~ 1/2 insert d tb with 
          | node tb1 c tb2 -> if coin
            then ~ 1/2 node (node tab ab tb1) d (node tb2 b tc)
            else       node tab ab (node (node tb1 c tb2) b tc)
    else match tab with
      | leaf         -> node (node leaf d leaf) ab tbc
      | node ta a tb -> if a <= d
        then match ~ 1/2 insert d tb with 
          | node tb1 c tb2 -> if coin
            then ~ 1/2 node (node ta a tb1) c (node tb2 ab tbc)
            else       node (node ta a (node tb1 c tb2)) ab tbc
        else match ~ 1/2 insert d ta with 
          | node ta1 c ta2 -> if coin
            then ~ 1/2 node ta1 c (node ta2 a (node tb ab tbc))
            else       node (node (node ta1 c ta2) a tb) ab tbc

delete_min z t = match t with
  | leaf          -> (leaf, z)
  | node tab b tc -> match tab with
    | leaf         -> (tc, b)
    | node ta a tb -> match ta with
      | leaf -> (node tb b tc, a)
      | ta   -> match ~ 1/2 delete_min z ta with
        | (t1, m) -> if coin
          then ~ 1/2 (node t1 a (node tb b tc), m)
          else       (node (node t1 a tb) b tc, m)
\end{lstlisting}
\caption{Module \inline{RandSplayHeap}.}
\label{fig:appendix-randsplayheap}
\end{figure}
\begin{figure}[t]
\begin{lstlisting}
meld h1 h2 = match h1 with
  | leaf             -> h2
  | node h1l h1x h1r -> match h2 with
    | leaf             -> (node h1l h1x h1r)
    | node h2l h2x h2r -> if h1x > h2x
      then if coin
        then (node (~ meld h2l (node h1l h1x h1r)) h2x h2r)
        else (node h2l h2x (~ meld h2r (node h1l h1x h1r)))
      else if coin
        then (node (~ meld h1l (node h2l h2x h2r)) h1x h1r)
        else (node h1l h1x (~ meld h1r (node h2l h2x h2r)))

insert x h = (meld (node leaf x leaf) h)

delete_min z h = match h with
  | leaf       -> (leaf, z)
  | node l x r -> ((meld l r), x)
\end{lstlisting}
\caption{Module \inline{RandMeldableHeap}.}
\label{fig:appendix-randmeldableheap}
\end{figure}
\begin{figure}[t]
\begin{lstlisting}
insert d t = match t with
  | leaf       -> node leaf d leaf
  | node l a r -> if coin 
    then node (~ insert d l) a r
    else node l a (~ insert d r)

delete z d t = match t with
  | node l a r -> if a == d
    then match l with
      | leaf -> r
      | l    -> match ~ delete_max z l with
        | (ll, m) -> node ll m r
    else if coin 
      then ~ (delete z d l)
      else ~ (delete z d r)

contains d t = match t with
  | leaf       -> false
  | node l a r -> if a == d
    then true
    else if coin 
      then ~ contains d l
      else ~ contains d r

delete_max z t = match t with
  | leaf         -> (leaf, z)
  | node cl c cr -> match cr with
    | leaf         -> (cl, c)
    | node bl b br -> match br with
      | leaf -> ((node cl c bl), b)
      | br   -> match ~ delete_max z br with
        | (t1, m) -> match t1 with
          | leaf         -> (leaf, z)
          | node al a ar -> (node (node (node cl c bl) b al) a ar, m)
\end{lstlisting}
\caption{Module \inline{CoinSearchTree}.}
\label{fig:appendix-coinsearchtree}
\end{figure}
\begin{figure}[t]
\begin{lstlisting}
descend t = match t with
  | leaf       -> leaf
  | node l a r -> if coin
    then node (~ descend l) a r
    else node l a (~ descend r)
\end{lstlisting}
\caption{Module \inline{RandTree}.}
\label{fig:appendix-tree}
\end{figure}

%% file: rules/typ/leaf.tex
\infer[\ruleleaf]{
  \tjudge{\varnothing}{Q+K}{\flstc{leaf}}{\TreeShort}{Q'}
}{%
  \forall c \geqslant 2 \ q_{(c)} = \sum_{a+b=c} q'_{(a,b)}
  &
  K = q'_\ast
}

%% file: rules/typ/node.tex
\infer[\rulenode]{
  \tjudge{\typed{x_1}{\TreeShort},\typed{x_2}{\BaseShort},\typed{x_3}{\TreeShort}}{Q}{\flsttree{x_1}{x_2}{x_3}}{\TreeShort}{Q'}
}{%
  q_1 = q_2 = q'_\ast
  &
  q_{(1,0,0)} = q_{(0,1,0)} = q'_\ast
  &
  q_{(a,a,c)} = q'_{(a,c)}
}

%% file: rules/typ/cmp.tex
\infer[\rulecmp]{
  \tjudge{\typed{x_1}{\alpha}, \typed{x_2}{\alpha}}{Q}{x_1 \circ x_2}{\BoolShort}{Q}
}{%
  \circ \in \{ <, >, = \}
}

%% file: rules/typ/var.tex
\infer[\rulevar]{
  \tjudge{\typed{x}{\alpha}}{Q}{x}{\alpha}{Q}
}{%
  \text{$x$ a variable}
}

%% file: rules/typ/matchpair.tex
\infer[\rulematchpair]{
    \tjudge{\Gamma, \typed{x}{\TypeA_1 \times \TypeA_2}}{Q}{\flstk{match }x\flstk{ with }
    \flst{| } \text{\flstpair{$x_1$}{$x_2$}} \flst{ -> }e}{\TypeB}{Q'}
}{%
\tjudge{\Gamma, \typed{x_1}{\TypeA_1}, \typed{x_2}{\TypeA_2}}{Q}{e}{\TypeB}{Q'}
          &     
\text{for at most one $i$, $\alpha_i = \TreeShort$}
}

%% file: rules/typ/lettreecf.tex
\infer[\rulelettreecf]{
  \tjudge{\Gamma, \Delta}{Q}{\flstk{let }x\flst{ = }e_1\flstk{ in }e_2}{\beta}{Q'}
}{%
  \begin{minipage}[b]{80ex}
    \centering
    $p_i = q_i$ \quad 
    $p_{(\veca,c)} = q_{(\veca, \vec{0}, c)}$ \quad
    $r_j = q_{m+j}$ \quad $r_{k+1} = p'_\ast$ \quad
    $r_{(\vec{0},d,e)} = p'_{(d,e)}$\\[1ex]
    $\forall \vec{b} \ne \vec{0} \left( r_{(\vec{b},0,c)} = q_{(\vec{0},\vec{b},c)} \right)$
    \\[1ex]
    $\forall {\vec{b} \ne \vec{0}}, {\vec{a} \ne \vec{0}}, d \neq 0 \
    \left( q_{(\veca,\vecb,c)} = \sum_{(d,e)} p^{(\vecb,d,e)}_{(\veca,c+\max{\{-e,0\}})} \right)$
    \\[1ex]
    $\forall \vec{b} {\ne \vec{0}}, {d \ne 0 } \
    \left(
    r_{(\vec{b},d,e)} = {p'}^{(\vecb,d,e)}_{(d,\max{\{e,0\}})}
    \wedge
    \forall (d',e') \neq (d,\max{\{e,0\}}) \ \left(
    {p'}^{(\vecb,d,e)}_{(d',e')} = 0 \right) \wedge {} \right.$\\[1ex]
    $\left. {} \land
    \sum_{(\veca,c)} p^{(\vecb,d,e)}_{(\veca,c)} \geqslant
    {p'}^{(\vecb,d,e)}_{(d,\max{\{e,0\}})} \land
    \forall \veca \ne \vec{0} \
    \left( p^{(\vecb,d,e)}_{(\veca,c)} \neq 0 \rightarrow  {p'}^{(\vecb,d,e)}_{(d,\max{\{e,0\}})} \leqslant p^{(\vecb,d,e)}_{(\veca,c)} \right) \right)$\\[2ex]
    $\tjudge{\Gamma}{P}{e_1}{\TreeShort}{P'}$
    \hfill
    $\forall {\vecb \ne \vec{0},d \ne 0 } \ \left( \tjudgecfndt{\Gamma}{P^{(\vecb,d,e)}}{e_1}{\TreeShort}{{P'}^{(\vecb,d,e)}} \right)$
    \\[1ex]
    $\tjudge{\Delta, \typed{x}{\TreeShort}}{R}{e_2}{\beta}{Q'}$
  \end{minipage}
}

%% file: rules/typ/wvar.tex
\infer[\rulewvar]{
  \tjudge{\Gamma, \typed{x}{\alpha}}{Q}{e}{\beta}{Q'}
}{%
  \tjudge{\Gamma}{R}{e}{\beta}{Q'}
  &
  r_i = q_i
  &
  r_{(\veca,b)} = q_{(\veca,0,b)}
}

%% file: rules/typ/share.tex
\infer[\ruleshare]{
  \tjudge{\Gamma, \typed{z}{\alpha}}{\share{Q}}{e[z,z]}{\beta}{Q'}
}{%
  \tjudge{\Gamma, \typed{x}{\alpha}, \typed{y}{\alpha}}{Q}{e[x,y]}{\beta}{Q'}
}

%% file: rules/typ/shift.tex
\infer[\ruleshift]{
  \tjudge{\Gamma}{Q+K}{e}{\alpha}{Q'+K}
}{%
  \tjudge{\Gamma}{Q}{e}{\alpha}{Q'}
  &
  K \geqslant 0
}

%% file: rules/typ/w.tex
\infer[\rulew]{%
  \tjudge{\Gamma}{Q}{e}{\alpha}{Q'}
}{%
  \tjudge{\Gamma}{P}{e}{\alpha}{P'}
  &
  \potential{\Gamma}{P} \leqslant \potential{\Gamma}{Q}
  &
  \potential{\Gamma}{P'} \geqslant \potential{\Gamma}{Q'}
}

%% file: rules/typ/pair.tex
\infer[\rulepair]{
  \tjudge{\typed{x_1}{\alpha_1}, \typed{x_2}{\alpha_2}}{Q}{\flstpair{x_1}{x_2}}{\alpha_1 \times \alpha_2}{Q'}
}{%
\text{for at most one $i$, $\alpha_i = \TreeShort$}
&
q_i = q'_\ast
&
q_{(a,c)} = q'_{(a,c)}
}

%% file: rules/typ/match.tex
\infer[\rulematch]{
    \tjudge{\Gamma, \typed{x}{\TreeShort}}{Q}{\flstk{match }x\flstk{ with }
    \flst{| }\flstc{leaf}\flst{ -> }e_1
    \flst{| }\flsttree{x_1}{x_2}{x_3}\flst{ -> }e_2}{\alpha}{Q'}
}{%
  \begin{minipage}[b]{25ex}
    $r_{(\veca,a,a,b)} = q_{(\veca,a,b)}$\\[1ex]
    $p_{(\veca,c)} = \sum_{a+b=c} q_{(\veca,a,b)}$\\[2ex]
    $\tjudge{\Gamma}{P+q_{m+1}}{e_1}{\alpha}{Q'}$
  \end{minipage}
  &
  \begin{minipage}[b]{45ex}
      $r_{m+1} = r_{m+2} = q_{m+1} \quad \quad \quad \quad q_i = r_i = p_i$
    \\[1ex]
    $r_{(\vec{0},1,0,0)} = r_{(\vec{0},0,1,0)} = q_{m+1}$
    \\[2ex]
    $\tjudge{\Gamma, \typed{x_1}{\TreeShort}, \typed{x_2}{\BaseShort}, \typed{x_3}{\TreeShort}}{R}{e_2}{\alpha}{Q'}$
  \end{minipage}
}

%% file: rules/typ/letnotree.tex
\infer[\ruleletnotree]{
  \tjudge{\Gamma, \Delta}{Q}{\flstk{let }x\flst{ = }e_1\flstk{ in }e_2}{\beta}{Q'}
}{%
  \begin{minipage}[b]{20ex}
    $p_i = q_i$\\[1ex]
    $p_{(\veca,c)} = q_{(\veca, \vec{0}, c)}$\\[2ex]
    $\tjudge{\Gamma}{P}{e_1}{\alpha}{P'}$
  \end{minipage}
  &
  \begin{minipage}[b]{25ex}
    $\forall \vec{b} \not= \vec{0} \ ( q_{(\vec{0},\vec{b},c)} = r_{(\vec{b},c)})$ \\[2ex]
    $\tjudge{\Delta, \typed{x}{\alpha}}{R}{e_2}{\beta}{Q'}$
  \end{minipage}
  &
  \begin{minipage}[b]{25ex}
    $r_{(\vec{0},c)} = p'_{(c)} \quad r_j = q_{m+j}$\\[2ex]
    $\alpha \not= \TreeShort$
  \end{minipage}
}

%% file: rules/typ/ite.tex
\infer[\ruleite]{
  \tjudge{\Gamma, \typed{x}{\BoolShort}}{Q}{\flstk{if }x\flstk{ then }e_1\flstk{ else }e_2}{\alpha}{Q'}
}{%
  \tjudge{\Gamma}{Q}{e_1}{\alpha}{Q'}
  &
  \tjudge{\Gamma}{Q}{e_2}{\alpha}{Q'}
}

%% file: rules/typ/ticking.tex
%
\infer[\ruletick]{%
  \tjudge{\Gamma}{Q+\sfrac{a}{b}}{\tick[$a$][$b$]{e}}{\alpha}{Q'}
}{%
  \tjudge{\Gamma}{Q}{e}{\alpha}{Q'}
}

%% file: rules/typ/itecoin.tex
%
\infer[\rulecoin]{%
   \tjudge{\Gamma}{Q}{\flstk{if } \coin[$a$][$b$] \flstk{ then }e_1\flstk{ else }e_2}{\alpha}{Q'}
}{%
   \tjudge{\Gamma}{Q_1}{e_1}{\alpha}{Q'}
  &
  \tjudge{\Gamma}{Q_2}{e_2}{\alpha}{Q'}
  &
  p = \sfrac{a}{b}
  &
  Q = p \cdot Q_1 + (1 - p) \cdot Q_2
}

%% file: rules/typ/appticks.tex
\infer[\ruleapp]{%
  \tjudge{\typed{x_1}{\alpha_1},\dots,\typed{x_n}{\alpha_n}}{(P +  K \cdot Q)}{f(\seq{x})}{\beta}{(P' +  K \cdot Q')}
}{%
  \atypdcl{\alpha_1 \times \cdots \times \alpha_n}{P}{\beta}{P'} \in \FS(f)
  &
  \atypdcl{\alpha_1 \times \cdots \times \alpha_n}{Q}{\beta}{Q'} \in \FScf(f)
  &
  K \in \Qplus
}

%% file: rules/sem/tickB.tex
\infer
  {\evaln{\sigma}{n+1}{c+\prob{\mu} \cdot \sfrac{a}{b}}{\tick[$a$][$b$]{e}}{\mu}}
  {\evaln{\sigma}{n}{c}{e}{\mu}}

%% file: rules/typ/tickingast.tex
\infer[\ruletickast]{%
  \tjudge{\Gamma}{Q}{\tick[$a$][$b$]{e}}{\alpha}{Q'-\sfrac{a}{b}}
}{%
  \tjudge{\Gamma}{Q}{e}{\alpha}{Q'}
}